\documentclass[letterpaper,twocolumn,10pt]{article}
\usepackage{usenix,epsfig,endnotes}

\usepackage{mathptmx}
\usepackage{multirow}
\usepackage{graphicx}

\newcounter{subcopyrightbox@save}
\usepackage[font=bf]{caption}
\usepackage[caption=false]{subfig}
\usepackage[ruled,linesnumbered]{algorithm2e}
\usepackage{color, url}
\usepackage{xspace} 
 \usepackage{balance}
\usepackage{mathrsfs}
\usepackage{amsmath}
\usepackage{amsthm}

\newcommand{\eat}[1]{}

\newtheorem{theorem}{Theorem}
\newtheorem{definition}[theorem]{Definition}
\newtheorem{corollary}[theorem]{Corollary}

\newcommand{\myparatight}[1]{\smallskip\noindent{\bf {#1}:}~}

\begin{document}


\title{You are Who You Know and How You Behave: Attribute Inference Attacks via Users' Social Friends and Behaviors}

%

\author{
{\rm Neil Zhenqiang Gong}\\
ECE Department, Iowa State University \\
neilgong@iastate.edu
\and
{\rm Bin Liu}\\
MSIS Department, Rutgers University\\
BenBinLiu@gmail.com
} 

\maketitle

\begin{abstract}
We propose new privacy attacks to infer attributes (e.g., locations, occupations, and interests)  of online social network  users.  
Our attacks leverage seemingly innocent user information that is publicly available in online social networks to infer missing attributes of targeted users.
Given the increasing availability of (seemingly innocent) user information online, our results 
have serious implications for Internet privacy -- private attributes can be inferred from users'
publicly available data unless we take steps to protect users from such inference attacks.

To infer attributes of a targeted user, existing  inference attacks leverage either the user's publicly available social friends or the user's behavioral records (e.g., the webpages that the user has liked on Facebook, the apps that the user has reviewed on Google Play), but not both.  As we will show, such inference attacks achieve limited success rates. 
However, the problem
becomes \emph{qualitatively} different if we consider both social friends and behavioral records.  
To address this challenge, we develop a novel model to integrate social friends and behavioral records 
and design new attacks based on our model. We theoretically and experimentally
demonstrate the effectiveness of our attacks. For instance, 
we observe that, in a real-world large-scale dataset with 1.1 million users,
our attack can correctly infer \emph{the cities a user lived in} for 57\% of the users; 
via \emph{confidence estimation}, we are able to
increase the attack success rate to over 90\% 
if the attacker selectively attacks a half of the users. 
Moreover,  we show that our attack can 
correctly infer attributes for significantly more users than previous attacks.

\end{abstract}


\section{Introduction}
Online social networks (e.g., Facebook, Google+, and Twitter) have become
increasingly important platforms for users to interact with each other, process
information, and diffuse social influence.  A user in an online social network 
 essentially has a list of social friends, a digital record of behaviors, and a profile.    
For instance, behavioral records could be a list of pages liked or shared by the user 
on Facebook, or they could be a set of mobile apps liked or rated by the user in Google+ or Google Play. 
  A profile introduces the user's 
self-declared attributes such as majors, employers, and cities lived. 
To address users' privacy concerns, online social network operators provide users with 
 fine-grained privacy settings, e.g., a user could limit some attributes  to be accessible 
only to his/her friends.  Moreover, a user could also create an
account without providing any attribute information. \emph{As a result, an online social network
is a mixture of both public and private user information}.

One privacy attack of increasing interest revolves around these user attributes~\cite{He06, Lindamood09, Zheleva09,AttributeInferenceFriendPET10,gong2014joint,Mislove10, AgeInferenceInfocom, Labitzke13, weinsberg2012blurme,McCallum98,Chaabane12,kosinski2013private,Humbert13,gardar,Jurgens15ICWSM}.  
In this \emph{attribute inference attack}, an attacker aims to
propagate attribute information of social network users with publicly visible attributes
to  users with missing or incomplete
attribute data. 
 Specifically, the attacker could be any party (e.g., cyber criminal, online social network provider, advertiser, data broker, and surveillance agency) who has interests in users' private attributes. To perform such privacy attacks, the attacker only needs to collect publicly available data from online social networks. Apart from privacy risks, the inferred user attributes can also be  used 
  to perform various security-sensitive activities such as spear phishing~\cite{spearphishing} and attacking personal information based backup authentication~\cite{Gupta13}. Moreover, an attacker can leverage the inferred attributes to link online users across multiple sites~\cite{Bartunov12,goga2013large,stylometrylink,userlinkAcrossSitesWWW13} or with offline records (e.g., publicly available voter registration records)~\cite{kanonymity,offlineattack} to  form composite user profiles,  resulting in even bigger security and privacy risks.



Existing attribute inference attacks can be roughly classified into two categories, \emph{friend-based}~\cite{He06, Lindamood09, Zheleva09, AttributeInferenceFriendPET10, gong2014joint,Mislove10, AgeInferenceInfocom, Labitzke13,Humbert13,gardar,Jurgens15ICWSM} and \emph{behavior-based}~\cite{weinsberg2012blurme,McCallum98,Chaabane12,kosinski2013private}. Friend-based attacks are based on the intuition of \emph{you are who you know}. Specifically, they
aim to infer  attributes for a user using the publicly available user attributes of the user's friends (or all other users in the social network) and the social structure among them. The foundation of friend-based attacks is \emph{homophily}, meaning that two linked users share similar attributes~\cite{homopily01}. For instance,  if more than half of friends of a user major in Computer Science at a certain university, the user 
might also major in Computer Science at the same university with a high probability. Behavior-based attacks infer attributes for a user based on the public attributes of users that are similar to him/her, 
and the similarities between users are identified by using their behavioral data.
The intuition behind behavior-based attacks is \emph{you are how you behave}. In particular,  users with the same attributes have similar interests, characteristics, and cultures so that they have similar behaviors. For instance, if a user liked apps, books, and music tracks on Google Play that are similar to those liked by users originally from China,   the user might also be from China. Likewise, previous measurement study~\cite{xu2011identifying} found that some apps are only popular in certain cities,  implying the possibility of inferring cities a user lived in using the apps the user used or liked.

However, these inference attacks consider either social friendship structures or user behaviors, but not both, and thus they achieve limited inference accuracy as we will show in our experiments. Moreover,  the problem of inferring user attributes becomes qualitatively different if we consider  both social structures and user behaviors  because features derived from them differ from each other, show different sparsity, and are at different scales. We show in our evaluation that simply concatenating features from the two sources of information regresses the overall results and reduces attack success rates.


\myparatight{Our work} In this work, we aim to combine social structures and user behaviors to infer user attributes. 
To this end, we first propose a \emph{social-behavior-attribute (SBA)} 
network model to gracefully integrate  social structures, user behaviors, and user attributes in a unified framework.
Specifically, we add additional nodes to a social structure, each of which represents an attribute or a behavior; a link between a user and an attribute node represents that the user has the corresponding attribute, and that a user has a behavior is encoded by a link between the user and the corresponding behavior node. 

Second, we design a \emph{vote distribution attack (VIAL)} under the SBA model to perform attribute inference. Specifically, VIAL iteratively distributes a fixed vote capacity 
 from a \emph{targeted user} whose attributes we want to infer to all other users in the SBA network. A user receives a high vote capacity if the user and the targeted user are structurally similar in the SBA network, e.g., they have similar social structures and/or have performed similar behaviors. Then, each user votes for its attributes via dividing its vote capacity to them. We predict the target user to own attributes that receive the highest votes. 
 
Third, we  evaluate VIAL both theoretically and empirically; and we extensively compare VIAL with several previous attacks for inferring majors, employers, and locations using  a  large-scale dataset with 1.1 million users collected from Google+ and Google Play.  For instance, we observe that our attack can correctly infer the cities a user lived in for 57\% of the users; via \emph{confidence estimation}, we are able to increase the success rate to over 90\% if the attacker selectively attacks a half of the users. Moreover, we find that our attack VIAL substantially outperforms previous attacks. Specifically, for Precision, VIAL improves  upon friend-based attacks and behavior-based attacks by over 20\% and around 100\%, respectively. These results imply that an attacker can use our  attack to successfully infer private attributes of substantially more users than previous attacks. 

In summary, our key contributions are as follows:
\begin{itemize}
\item We propose the \emph{social-behavior-attribute (SBA)} network model to integrate social structures, user behaviors, and user attributes.

\item  We design the \emph{vote  distribution attack (VIAL)} under the SBA model to perform attribute inference. 

\item We demonstrate the effectiveness of VIAL both theoretically and empirically. Moreover, we observe that VIAL  correctly infers attributes for substantially more users than previous attacks via evaluations on a  large-scale dataset collected from Google+ and Google Play.
\end{itemize}

\section{Problem Definition and Threat Model}
 
\myparatight{Attackers} The attacker could be any party who has interests in user attributes. 
For instance, the attacker could be a cyber criminal, online social network provider, advertiser, data broker, or surveillance agency. 
Cyber criminals  can leverage user attributes to perform targeted social engineering attacks (now often referred to as spear phishing attacks~\cite{spearphishing}) and attacking personal information based backup authentication~\cite{Gupta13};
online social network providers and advertisers could use the user attributes for targeted advertisements; 
 data brokers make profit via selling the user attribute information to other parties such as advertisers, banking companies,
 and insurance industries~\cite{databroker}; and surveillance agency can use the attributes to identify users and monitor their activities.

\myparatight{Collecting publicly available social structures and behaviors} To perform attribute inference attacks, an attacker first needs to collect 
publicly available information. In particular, in our attacks, an attacker needs to collect social structures, user profiles, and user behaviors from online social networks. 
Such information  can be collected via writing web crawlers or leveraging APIs developed by the service providers. 
Next, we formally describe these publicly available information.

We use an undirected\footnote{Our attacks
can also be generalized to directed graphs.} graph $G_s=(V_s, E_s)$ to represent a
social structure,  where edges in $E_s$ represent social relationships  between the
 nodes in $V_s$. We denote by $\Gamma_{u,S}=\{v|(u,v)\in E_s\}$ as the 
set of social neighbors of $u$.
In addition to social network structure, we have behaviors and categorical
attributes for nodes. For instance, in our Google+ and Google Play dataset,
 nodes are
Google+ users, and edges represent friendship  between users; 
behaviors include the set of items 
(e.g., apps, books, and movies) that users rated or liked on Google Play; 
and node attributes are derived from user profile information and include
fields such as major, employer, and cities lived.  

We use binary representation for user behaviors. Specifically,
 we treat various objects 
(e.g., the Android app ``Angry Birds", the movie ``The Lord of the Rings", and the webpage ``facebook.com") as binary variables,
 and we denote by $m_b$ the total number of objects.
 Behaviors of a
node $u$ are then represented as a $m_b$-dimensional binary column vector
$\vec{b}_u$ with the $i${th} entry equal to $1$ when $u$ has performed a certain 
action on the $i${th} object (\emph{positive behavior}) and $-1$ when $u$ does not 
perform the action on it (\emph{negative behavior}). 
For instance, when we consider user review behaviors for Google+ users,  
objects could be items such as apps, books, and movies available in Google Play, 
and the action is \emph{review}; 
1 represents that the user reviewed the corresponding item 
and -1 means the opposite. 
For Facebook users,  objects could be webpages; 
1 represents that the user liked or shared the 
corresponding webpage and -1 means that the user did not.   
We denote by ${\bf B}=[\vec{b}_1\ \vec{b}_2 \cdot
\cdot \cdot \vec{b}_{n_s}]$ the behavior matrix for all nodes. 

We distinguish between \emph{attributes} and \emph{attribute values}. 
For instance, major, employer, and location are different attributes; 
and each attribute could have multiple attribute values, e.g.,
major could be Computer Science, Biology, or Physics. 
A user might own a few attribute values for a single attribute. 
For example, a user that studies  Physics for undergraduate education but 
chooses to pursue a Ph.D. degree in Computer Science has two values 
for the attribute major.  
Again,  we use a binary representation for each attribute value, 
and we denote the number of distinct attribute values as $m_a$.  
Then attribute information of a
node $u$ is represented as a $m_a$-dimensional binary column vector
$\vec{a}_u$ with the $i${th} entry equal to $1$ when $u$ has the $i${th}
attribute value (\emph{positive attribute}) and $-1$ when $u$ does not have it
(\emph{negative attribute}).
We denote by ${\bf A}=[\vec{a}_1\ \vec{a}_2 \cdot
\cdot \cdot \vec{a}_{n_s}]$ the attribute matrix for all nodes. 

\myparatight{Attribute inference attacks} Roughly speaking, an attribute inference attack is to infer the attributes of a set of targeted users using the collected publicly available information. Formally, we define an attribute inference attack as follows:

\begin{definition}[Attribute Inference Attack]
Suppose we are given $T=(G_s,{\bf A},{\bf B})$, which is a snapshot of a social network $G_s$ with a behavior matrix ${\bf B}$ and an attribute matrix ${\bf A}$, and a list of targeted users $V_{t}$ with social friends $\Gamma_{v,S}$ and binary behavior vectors $\vec{b}_v$ for all $v\in V_{t}$, the attribute inference attack is to infer the attribute vectors $\vec{a}_v$ for all $v\in V_{t}$.
\end{definition}

We note that a user setting the friend list to be private could also be vulnerable to inference attacks. 
This is because the user's friends could set their friend lists publicly available. 
The attacker can  collect a social relationship between two users  if 
at least one of them sets the friend list to be public. 
Moreover, we assume the users and the service providers are not taking other steps  
(e.g., obfuscating social friends~\cite{friendbasedDefense} or behaviors~\cite{weinsberg2012blurme,inferenceDefensePETS14}) 
to defend against inference attacks. 

\myparatight{Applying inferred attributes to link users across multiple online social networks and with offline records} 
We stress that an attacker could leverage our attribute inference attacks to further perform other attacks. 
For instance, a user might provide different attributes on different online social networks. 
Thus, an attacker could combine user attributes across multiple online social networks to better profile users,
and an attacker could leverage the inferred user attributes to do so~\cite{Bartunov12,goga2013large,stylometrylink,userlinkAcrossSitesWWW13}. 
Moreover, an attacker can further use the inferred user attributes to link online users with offline records (e.g., voter registration records)~\cite{kanonymity,offlineattack}, 
which results in even bigger security and privacy risks, e.g., more sophisticated social engineering attacks. 
We note that even if the inferred user attributes (e.g., major, employer) seem not private for some targeted users, 
an attacker could use them to link users across multiple online sites and with offline records.

\section{Social-Behavior-Attribute Framework}
We describe our \emph{social-behavior-attribute (SBA)} network model, which integrates social structures, user behaviors, and user attributes in a unified  framework. To perform our inference attacks,  an attacker needs to construct a SBA network from his/her collected publicly available social structures, user attributes, and behaviors.


Given a social network $G_s=(V_s,E_s)$ with $m_b$ behavior objects, 
 a behavior matrix ${\bf B}$, $m_a$ distinct attribute values, and an
attribute matrix ${\bf A}$,  
 we create an augmented
network by adding  $m_b$ additional nodes to $G_s$, with each node corresponding
 to a behavior object, and another $m_a$ additional nodes to $G_s$, 
with each additional node
corresponding to an attribute value. 
  For each node $u$ in $G_s$ with positive attribute $a$ or positive behavior $b$, 
we create an undirected link between $u$ and the additional 
node corresponding to $a$ or $b$ in the
augmented network. Moreover, we add the targeted users into the augmented 
network by connecting them to their friends and the additional nodes corresponding
 to their positive behaviors.  
We call this augmented network \emph{social-behavior-attribute (SBA)} network 
since it integrates the interactions among social
 structures, user behaviors, and user attributes.

Nodes in the SBA framework corresponding to nodes in $G_s$ or targeted users in $V_{t}$ 
are called \emph{social
nodes},  nodes representing behavior objects are called  \emph{behavior nodes},
and nodes representing attribute values are called \emph{attribute nodes}.
Moreover, we  use $S$,  $B$, and $A$ 
to represent the three types of nodes, respectively.  
Links between social nodes are called \emph{social links}, 
 links between
social nodes and behavior nodes are called \emph{behavior links},
and links between
social nodes and attribute nodes are called \emph{attribute links}. 
Note that there are no links between  behavior nodes and attribute nodes. 
Fig.~\ref{sab} illustrates an example SBA network, in which the two social nodes $u_5$ and $u_6$ correspond to two targeted users. The behavior nodes in this example correspond to Android apps, and 
a behavior link  represents that the corresponding user used the 
corresponding app. 
 Intuitively,
the SBA framework explicitly describes the sharing of  behaviors and attributes across social
nodes. 

\begin{figure}[!t]
\centering
\includegraphics[width=0.4 \textwidth]{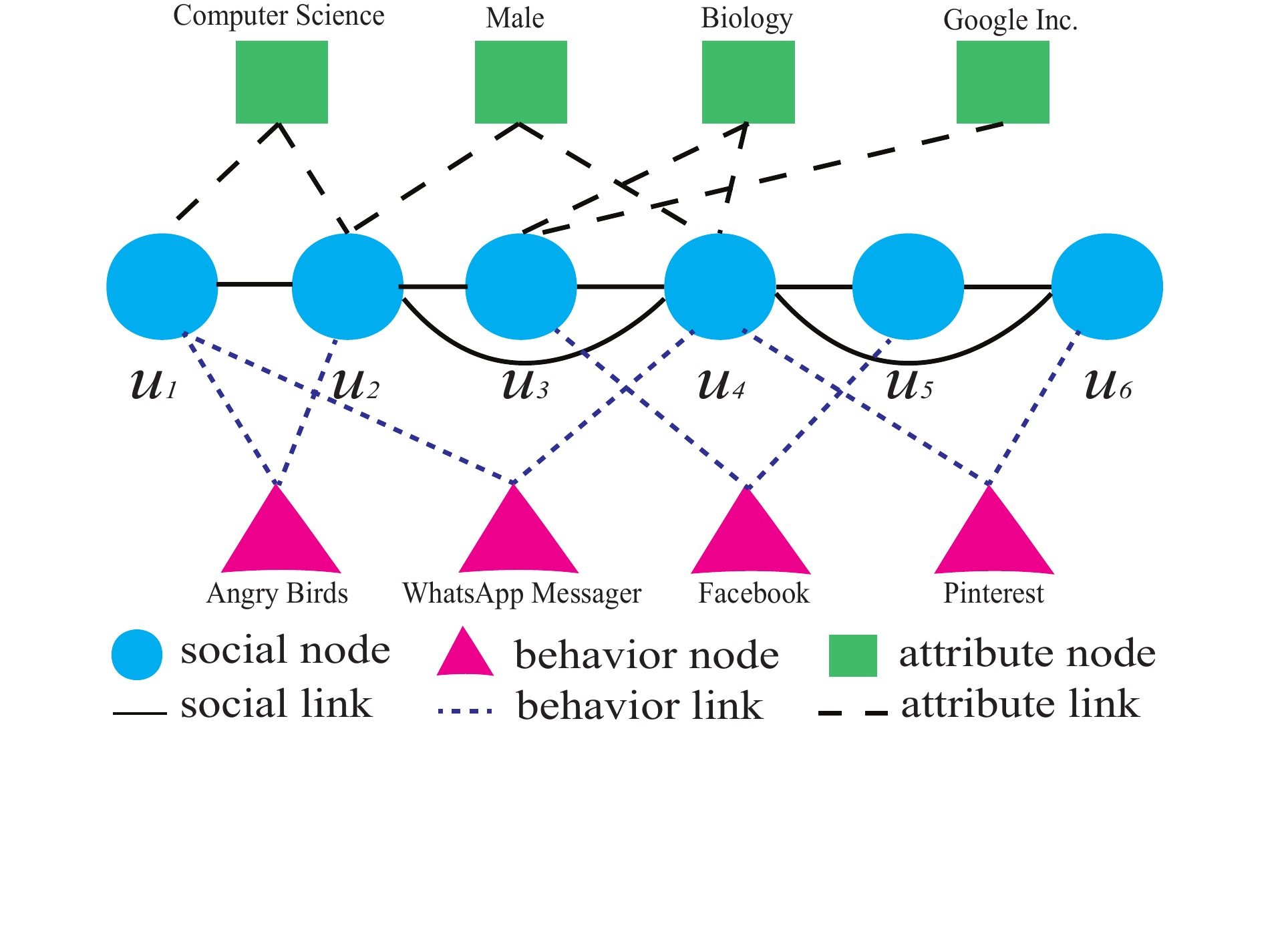} 
\vspace{-10mm}
\caption{Social-behavior-attribute network.}
\label{sab}
\end{figure}

We also place weights on various links in the SBA framework.
These link weights 
balance the influence of social links versus behavior links versus
attribute links.\footnote{In principle, 
we could also assign weights to nodes to incorporate their relative importance. 
However, our attack does not rely on node weights, so we do not discuss them.}  
For instance, weights on social links could represent the tie strengths between social nodes. 
Users with stronger tie strengths could be more likely to share the same attribute values. 
The weight on a behavior link could indicate the predictiveness of the behavior in terms of the user's attributes. In other words, a behavior link with a higher weight means that performing the corresponding behavior better predicts the attributes of the user. For instance, if we want to predict user gender, the weight of the link between a female user and  a mobile app
 tracking women's monthly periods could be larger than the weight of the link between a male user and the app.
 Weights on attribute links can represent the degree of affinity between  users and  attribute values. For instance, an attribute link connecting the user's hometown could have a higher weight than the attribute link connecting a city where the user once travelled. 
 We discuss how link weights can be learnt via machine learning in Section~\ref{app:weight}.


We denote a SBA network as $G=(V,E,w,t)$, where $V$ is the set of nodes, 
$n=|V|$ is the total number of nodes,  
$E$ is the set of links, $m=|E|$ is the total number of links, 
$w$ is a function that maps a link to its link weight, 
i.e., $w_{uv}$ is the weight of link $(u,v)$, and $t$ a function that maps
a node to its node type, i.e., $t_u$ is the node type of $u$.
For instance, $t_u= S$ means that $u$ is a social node.  
Additionally, for a given node $u$ in the SBA network, we denote by $\Gamma_{u}$, $\Gamma_{u,S}$, 
$\Gamma_{u,B}$, and $\Gamma_{u,A}$ respectively  
the sets of \emph{all neighbors}, \emph{social neighbors}, 
\emph{behavior neighbors}, and \emph{attribute neighbors} of $u$. 
Moreover, for links that are incident from $u$, we use $d_u$,   $d_{u,S}$,   $d_{u,B}$, and $d_{u,A}$ 
to denote the sum of weights
of all links, weights of links connecting social neighbors,
  weights of links connecting behavior neighbors, 
and weights of links connecting attribute neighbors, respectively.  More specifically, 
we have $d_u=\sum_{v\in \Gamma_u} w_{uv}$ and $d_{u,Y}=\sum_{v\in \Gamma_{u,Y}} w_{uv}$, where $Y=S, B, A$. 

Furthermore, we define two types of \emph{hop-2 social neighbors} of a social node $u$, which share common behavior neighbors or  attribute neighbors with $u$. In particular, a social node $v$ is called a \emph{behavior-sharing} social neighbor of $u$ if $v$ and $u$ share at least one common behavior neighbor. For instance, in Fig.~\ref{sab}, both $u_2$ and $u_4$ are behavior-sharing social neighbors of $u_1$. 
We denote the set of behavior-sharing social neighbors of $u$ as $\Gamma_{u,BS}$. Similarly, we denote the set of attribute-sharing social neighbors of $u$ as $\Gamma_{u,AS}$. Formally, we have $\Gamma_{u,BS}$=$\{v| t(v) =S\ \&\ \Gamma_{v,B} \cap \Gamma_{u,B} \neq \emptyset\}$ and $\Gamma_{u,AS}$=$\{v| t(v) =S\ \&\ \Gamma_{v,A} \cap \Gamma_{u,A} \neq \emptyset\}$. We note that our definitions of $\Gamma_{u,BS}$ and $\Gamma_{u,AS}$ also include the social node $u$ itself. 
These notations will be useful in describing our  attack.

\section{Vote Distribution Attack (VIAL)}
\subsection{Overview} 
Suppose we are given a SBA network $G$ which also 
includes the social structures and 
behaviors of the targeted users, our goal is to infer 
attributes for every targeted user. 
Specifically, for each targeted user $v$, 
we compute the similarity between $v$ and each attribute value,
and then we predict that $v$ owns the attribute values that have the highest
similarity scores. In a high-level abstraction, VIAL works in two phases. 
\begin{itemize}
\item {\bf Phase I.} VIAL iteratively distributes a fixed vote capacity from the targeted user $v$ to the rest of  users in Phase I. The intuitions are that a user receives a high vote capacity if the user and the targeted user are structurally  similar in the SBA network (e.g., share common friends and behaviors), and that the targeted user is more likely to have the attribute values belonging to users with higher vote capacities. 
After Phase I, we obtain a vote capacity vector $\vec{s}_v$, where  $\vec{s}_{vu}$ is the vote capacity of user $u$. 
 \item {\bf Phase II.} Intuitively, if a user with a certain vote capacity has more attribute values, then, according to the information of this user alone,  the likelihood of each of these attribute values belonging to the targeted user decreases. Moreover, an attribute value should receive more votes if more users with higher vote capacities have the attribute value. Therefore, in Phase II, each social node votes for its attribute values via dividing its vote capacity among them, and each attribute value sums the vote capacities that are divided  to it by its social neighbors. We treat the summed vote capacity of an attribute value as its similarity with $v$. Finally, we predict  $v$ has the attribute values that receive the highest votes. 
\end{itemize}


\subsection{Phase I}
In Phase I, VIAL iteratively distributes a fixed vote capacity from the targeted user $v$ to the rest of  users. 
We denote by $\vec{s}^{(i)}_v$ the vote capacity vector in the $i$th iteration, 
where  $\vec{s}^{(i)}_{vu}$ is the vote capacity of node $u$ in the $i$th iteration. Initially,  $v$ has a vote capacity $|V_s|$ 
 and all other social nodes have vote capacities of 0. Formally, we have:
\begin{align}
\vec{s}^{(0)}_{vu}=
\begin{cases} 
|V_s| & \text{ if } u=v  \\ 
0 &\text{ otherwise}
\end{cases}
\label{initial}
\end{align}


\begin{figure}[!t]
\centering
\includegraphics[width=0.45 \textwidth]{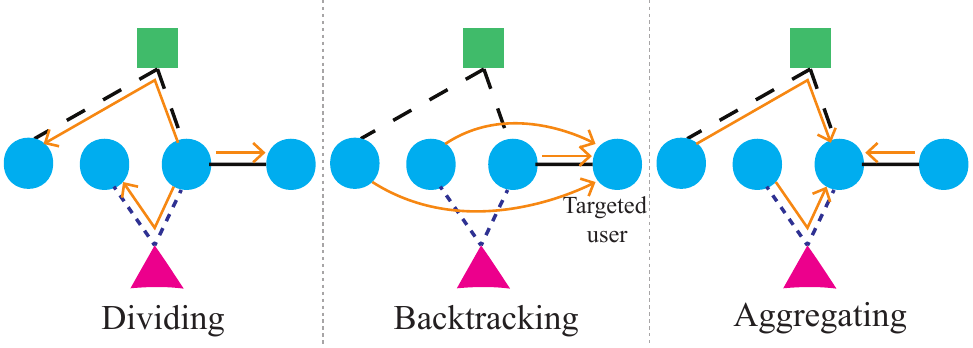} 
\caption{Illustration of our three local rules.}
\label{heu}
\end{figure}

In each iteration, VIAL applies three local rules. They are \emph{dividing}, \emph{backtracking}, and \emph{aggregating}.
Intuitively, if a user $u$ has more (hop-2) social neighbors, then each neighbor  could receive less vote capacity from $u$. 
Therefore, our dividing rule splits a social node's vote capacity to its 
social neighbors and hop-2 social neighbors. 
The backtracking  rule takes a portion 
of every social node's vote capacity and assigns them back to the targeted user $v$, 
which is based on the intuition that social nodes that are closer to $v$ in the SBA network 
are likely to be more similar to $v$ and should get more vote capacities.  
A user could have a higher vote capacity if it is linked to more social neighbors and  hop-2 social neighbors with higher vote capacities. Thus, 
for each user $u$,  the aggregating rule collects
 the vote capacities that are shared to $u$ by its social neighbors and hop-2 social neighbors.
Fig.~\ref{heu} illustrates the three local rules.  
Next, we elaborate the three local rules.

\myparatight{Dividing}
A social node $u$ could have social neighbors, behavior-sharing social neighbors, and attribute-sharing social neighbors. 
To distinguish them, 
we use three weights $w_S$, $w_{BS}$, and $w_{AS}$ to 
represent the shares of them,  respectively. 
For instance, 
the total vote capacity shared to social neighbors of $u$  in the $t$th iteration
 is $\vec{s}^{(i-1)}_{vu}\times \frac{w_{S}}{w_S+w_{BS}+w_{AS}}$.  
Then we further divide the vote capacity among each type 
of neighbors according to their link weights. We define $I_{u,Y}=1$ if the set of neighbors $\Gamma_{u,Y}$ is non-empty, otherwise $I_{u,Y}=0$, where $Y=S,BS,AS$. 
 The variables $I_{u,S}$, $I_{u,BS}$, and $I_{u,AS}$ 
are used to consider the scenarios where $u$ does not have some type(s) of neighbors, 
in which $u$'s vote capacity is divided among less than three types of social neighbors. For convenience, we denote $w_T=w_SI_{u,S}+w_{BS}I_{u,BS}+w_{AS}I_{u,AS}$.
\begin{itemize}
\item {\bf Social neighbors.}  A social neighbor $x\in \Gamma_{u,S}$ receives a higher vote capacity from $u$ if their link weight (e.g., tie strength) is higher.  
Therefore, we model the vote capacity $p^{(i)}_v(u,x)$ that is divided to $x$ by $u$ in the $i$th iteration as: 
\begin{align}
p^{(i)}_v(u,x)=\vec{s}^{(i-1)}_{vu}\cdot \frac{w_S}{w_T}\cdot \frac{w_{ux}}{d_{u,S}},
\end{align}
where $d_{u,S}$ is the summation of weights of social links that are incident from $u$. 

\item {\bf Behavior-sharing social neighbors.}  A behavior-sharing social neighbor $x\in \Gamma_{u,BS}$ receives a higher vote capacity from $u$ if they share more behavior neighbors with higher predictiveness. Thus, we model vote capacity $q^{(i)}_v(u,x)$ that is divided to $x$ by $u$ in the $i$th iteration as: 
\begin{align}
q^{(i)}_v(u,x)=\vec{s}^{(i-1)}_{vu} \cdot \frac{w_{BS}}{w_T} \cdot w_B(u,x),
\end{align}
where $w_B(u,x)=\sum_{y\in \Gamma_{u,B} \cap\Gamma_{x,B}} \frac{w_{uy}}{d_{u,B}} \cdot \frac{w_{xy}}{d_{y,S}}$, representing the overall share of vote capacity that $u$ divides to $x$ because of their common behavior neighbors. Specifically, $\frac{w_{uy}}{d_{u,B}}$ characterizes the fraction of vote capacity $u$ divides to the behavior neighbor $y$ and $\frac{w_{xy}}{d_{y,S}}$ characterizes the fraction of vote capacity $y$ divides to $x$. Large weights of $w_{uy}$ and $w_{xy}$  indicate $y$ is a predictive behavior of the attribute values of $u$ and $x$, and having more such common behavior neighbors make $x$ share more vote capacity from $u$.

\item {\bf Attribute-sharing social neighbors.}  An attribute-sharing social neighbor $x\in \Gamma_{u,AS}$ receives a higher vote capacity from $u$ if they share more attribute neighbors with higher degree of affinity. Thus, we model vote capacity $r^{(i)}_v(u,x)$ that is divided to $x$ by $u$ in the $i$th iteration as: 
\begin{align}
r^{(i)}_v(u,x)=\vec{s}^{(i-1)}_{vu} \cdot \frac{w_{AS}}{w_T} \cdot w_A(u,x),
\end{align}
where $w_A(u,x)=\sum_{y\in \Gamma_{u,A} \cap\Gamma_{x,A}} \frac{w_{uy}}{d_{u,A}} \cdot \frac{w_{xy}}{d_{y,S}}$, representing the overall share of vote capacity that $u$ divides to $x$ because of their common attribute neighbors. Specifically, $ \frac{w_{uy}}{d_{u,A}}$ characterizes the fraction of vote capacity $u$ divides to the attribute neighbor $y$ and $\frac{w_{xy}}{d_{y,S}}$ characterizes the fraction of vote capacity $y$ divides to $x$. Large weights of $w_{uy}$ and $w_{xy}$  indicate $y$ is an attribute value with a high degree of affinity, and having more such common attribute values make $x$ share more vote capacity from $u$.
\end{itemize}

We note that a social node $x$ could be multiple types of social neighbors of $u$ (e.g., $x$ could be social neighbor and behavior-sharing social neighbor of $u$), in which $x$ receives multiple shares of vote capacity from $u$ and we sum them as $x$'s final share of vote capacity.

    \begin{figure*}
\begin{align}
&\text{\bf Our aggregating rule to compute the new vote capacity } \vec{s}^{(i)}_{vu}\  \text{\bf for $u$:}  \nonumber \\
&\vec{s}^{(i)}_{vu}=
\begin{cases} 
(1-\alpha)(\sum_{x\in \Gamma_{u,S}} p^{(i)}_v(x,u) + \sum_{x\in \Gamma_{u,BS}} q^{(i)}_v(x,u) + \sum_{x\in \Gamma_{u,AS}} r^{(i)}_v(x,u)) & \text{ if } u\neq v  \\ 
(1-\alpha)(\sum_{x\in \Gamma_{u,S}} p^{(i)}_v(x,u) + \sum_{x\in \Gamma_{u,BS}} q^{(i)}_v(x,u) + \sum_{x\in \Gamma_{u,AS}} r^{(i)}_v(x,u)) + \alpha |V_s| &\text{ otherwise}
\end{cases}
\label{aggre}
\end{align}
\vspace{-4mm}
 \end{figure*}
 \begin{figure*}
\begin{align}
&\text{\bf Our dividing matrix:}  \nonumber \\
&{\bf M}_{ux}=
\begin{cases} 
\delta_{ux,S}\cdot \frac{w_{S}}{w_T}\cdot \frac{w_{ux}}{d_{u,S}} +\delta_{ux,BS} \cdot \frac{w_{BS}}{w_T} \cdot w_B(u,x)  + \delta_{ux,AS} \cdot \frac{w_{AS}}{w_T} \cdot w_A(u,x)  & \text{ if } x\in \Gamma_{u,S} \cup \Gamma_{u,BS} \cup \Gamma_{u,AS}    \\ 
0 &\text{ otherwise,}
\end{cases}
\label{matrix}
\end{align}
\hspace{21mm} where $\delta_{ux,Y}=1$ if  $x \in \Gamma_{u,Y}$, {otherwise } $\delta_{ux,Y}=0$, $Y=S,BS,AS$.
\end{figure*}

\myparatight{Backtracking} For each social node $u$, 
the backtracking rule takes a portion $\alpha$ 
of $u$'s vote capacity back to the targeted user $v$. 
Specifically, the vote capacity backtracked to $v$ from $u$ 
is $\alpha \vec{s}^{(i-1)}_{vu}$. Considering backtracking, 
 the  vote capacity divided to the social neighbor $x$ of $u$
 in the dividing step is modified as $(1-\alpha)p^{(i)}_v(u,x)$. Similarly, 
 the   vote capacities divided to a behavior-sharing social neighbor and an 
 attribute-sharing social neighbor $x$ are modified as $(1-\alpha)q^{(i)}_v(u,x)$ 
 and $(1-\alpha)r^{(i)}_v(u,x)$, respectively. 
We call the parameter $\alpha$ \emph{backtracking strength}. 
A larger  backtracking strength enforces more vote capacity to be distributed 
among the social nodes that are closer to $v$ in the SBA network.  
$\alpha=0$ means no backtracking. We will show that, via both theoretical and empirical evaluations, 
VIAL achieves better accuracy with backtracking. 

We can verify that $\sum_{x\in \Gamma_{u,S}} \frac{w_{ux}}{d_{u,S}}= 1$, $\sum_{x\in \Gamma_{u,BS}} w_B(u,x) = 1$, and $\sum_{x\in \Gamma_{u,AS}} w_A(u,x) = 1$ for every user $u$ in the dividing step. In other words, every user divides all its vote capacity to its neighbors  (including the user itself if the user has hop-2 social neighbors). Therefore, the total vote capacity keeps unchanged in every iteration, and the vote capacity that is backtracked to the targeted user is $\alpha |V_s|$.

\myparatight{Aggregating}
The aggregating rule computes 
a new vote capacity for $u$ 
by aggregating the vote capacities that are divided to $u$ by its neighbors in the $i$th iteration. 
For the targeted user $v$, 
we also collect the vote capacities that are backtracked from all social nodes.
Formally, our aggregating rule is represented as Equation~\ref{aggre}.

%
%

\begin{algorithm}[!t] 
\DontPrintSemicolon 
\KwIn{$G=(V,E,w,t)$, ${\bf M} $, $v$, $\epsilon$, and $\alpha$.} 
\KwOut{$\vec{s}_v$.} 
\Begin{
//Initializing the vote capacity vector.\;
\For{$u\in V_s$}{
\eIf{$u=v$}{
	$\vec{s}^{(0)}_{vu}\longleftarrow |V_s|$ \;	
}{
	$\vec{s}^{(0)}_{vu}\longleftarrow 0$ \;	
}
}
$error \longleftarrow 1$\;
\While{$error > \epsilon$}{
$\vec{s}^{(i)}_{v}\longleftarrow \alpha \vec{e}_v  + (1-\alpha) {\bf M}^T \vec{s}^{(i-1)}_{v}$ \;	
$error \longleftarrow |\vec{s}^{(i)}_{v} - \vec{s}^{(i-1)}_{v}|/|V_s|$ \;	
}

\Return $\vec{s}^{(i)}_{v}$\;
} 
\caption{Phase I of VIAL} 
\label{alg:VIAL}
\end{algorithm}

\myparatight{Matrix representation}
We derive the Phase I of our attack using matrix terminologies,
which makes it easier to iteratively compute the vote capacities. 
Towards this end, we define a
\emph{dividing matrix} ${\bf M}\in R^{|V_s|\times |V_s|}$, which is formally represented in Equation~\ref{matrix}. 
The dividing matrix  encodes the dividing rule. 
Specifically, $u$ divides ${\bf M}_{ux}$ fraction of
 its vote capacity to the neighbor
$x$ in the  dividing step. Note that  ${\bf M}$ includes the dividing rule for all three types of social neighbors. 
With the dividing matrix ${\bf M}$, we can represent the backtracking and aggregating rules in the $i$th iteration
as follows:
\begin{align}
\vec{s}^{(i)}_v=\alpha \vec{e}_v + (1-\alpha) {\bf M}^T \vec{s}^{(i-1)}_v, 
\label{iter}
\end{align}
where $\vec{e}_v$ is a vector with the $v$th entry equals $|V_s|$ and all other
entries equal 0, and ${\bf M}^T$ is the transpose of ${\bf M}$. 

Given an initial vote capacity  vector specified in Equation~\ref{initial}, we iteratively apply
 Equation~\ref{iter} until the difference between the vectors in two consecutive iterations 
is smaller than a predefined threshold. Algorithm~\ref{alg:VIAL} shows Phase I of our attack. 

\subsection{Phase II}
In Phase I, we obtained a vote capacity for each user. On one hand, the targeted user could be  more likely to share attribute values with the users with higher vote capacities. 
On the other hand, if a user has more attribute values, then the likelihood of each of these attribute values belonging to the targeted user could be smaller. For instance, if a user with a high vote capacity once studied in more universities for undergraduate education, then according to this user's information alone,  the likelihood of the targeted user studying in each of those universities could be smaller. 

Moreover, among a user's attribute values, 
an attribute value that has a higher degree of affinity (represented by the weight of the corresponding attribute link) with the user could be more likely to be an attribute value of the targeted user. For instance, suppose a user once lived in two cities, one of which is the user's hometown while the other of which is a city where the user once travelled; the user has a high vote capacity because he/she is structurally close (e.g., he/she shares many common friends with the targeted user) to the targeted user; then the targeted user is more likely to be from the hometown of the user than from the city the user once travelled. 

Therefore, to capture these observations, we divide the vote capacity of a user to its attribute values in proportion to the weights of its attribute links; and each attribute value sums the vote capacities that are divided to it by the users having the attribute value. Intuitively, an attribute value receives more votes if more users with higher vote capacities link to the attribute value via links with higher weights. Formally, we have 
\begin{align}
\vec{t}_{va}=\sum_{u\in \Gamma_{a, S}} \vec{s}_{vu} \cdot \frac{w_{au}}{d_{u,A}},
\label{phaseii}
\end{align}
where $\vec{t}_{va}$ is the final votes of the attribute value $a$, $\Gamma_{a, S}$ is the set of users who have the attribute value $a$, $d_{u,A}$ is the sum of weights of attribute links that are incident from $u$.  

We treat the summed votes of an attribute value as  its similarity with $v$. Finally, we predict  $v$ has the attribute values that receive the highest votes.

 

 \subsection{Confidence Estimation}
 For a targeted user, 
a \emph{confidence estimator} takes the final votes for all attribute values as an input and produces a confidence score. 
A higher confidence score means that attribute inference for the targeted user is more trustworthy. 
We design a  confidence estimator based on clustering techniques. 
 A targeted user could have multiple attribute values for a single attribute, and our attack could produce close votes for these attribute values. 
Therefore, we design a  confidence estimator called \emph{clusterness} for our attack. 
Specifically, we first use a clustering algorithm (e.g., k-means~\cite{kmeans}) to group the votes that our attack produces for all candidate attribute values into two clusters. 
 Then we compute the average vote in each cluster, and the clusterness is the difference between the two average votes. 
The intuition of our  clusterness is that if our attack successfully infers the targeted user's attribute values, 
there could be a cluster of attribute values whose votes are significantly higher than other attribute values'. 

Suppose the attacker chooses a confidence threshold and only predicts attributes for targeted users whose confidence scores 
are higher than the threshold. Via setting a larger confidence threshold, the attacker  
will attack less targeted users but could achieve a higher success rate.
In other words, an attacker 
can balance between  the success rates and the number of targeted users to attack via confidence estimation.

\section{Theoretical Analysis}
We analyze the convergence of VIAL and derive the analytical forms of vote capacity vectors, discuss the importance of the backtracking rule, and analyze the complexity of VIAL. 
 
\subsection{Convergence and Analytical Solutions} 


We first show that for any backtracking strength $\alpha \in (0, 1]$, the vote capacity vectors converge. 
 \setcounter{theorem}{0} 
\begin{theorem}
\label{converge}
For any backtracking strength $\alpha \in (0, 1]$, the vote capacity vectors $\vec{s}^{(0)}_v$, $\vec{s}^{(1)}_v$, $\vec{s}^{(2)}_v$, $\cdots$ converge, and the converged vote capacity vector is $\alpha(I-(1-\alpha) {\bf M}^T)^{-1} \vec{e}_v$. Formally, we have:
\begin{align}
\vec{s}_v = \lim_{i\to \infty}\vec{s}^{(i)}_v=\alpha(I-(1-\alpha) {\bf M}^T)^{-1} \vec{e}_v,
\end{align}
where $I$ is an identity matrix and $(I-(1-\alpha) {\bf M}^T)^{-1}$ is the inverse of $(I-(1-\alpha) {\bf M}^T)$.
\end{theorem}
\begin{proof}
See Appendix~\ref{app:converge}.
\end{proof}

Next, we analyze the convergence of VIAL and the analytical form of the vote capacity vector  when the backtracking strength $\alpha = 0$. 
\begin{theorem}
\label{nobacktracking}
When $\alpha=0$ and the SBA network is connected, the vote capacity vectors $\vec{s}^{(0)}_v$, $\vec{s}^{(1)}_v$, $\vec{s}^{(2)}_v$,  $\cdots$ converge, and the converged vote capacity vector is proportional to the unique stationary distribution of the Markov chain whose transition matrix is {\bf M}. Mathematically, the converged vote capacity vector $\vec{s}_v$ can be represented as:
\begin{align}
\vec{s}_v = |V_s| \vec{\pi},
\end{align}
where $\vec{\pi}$ is the unique stationary distribution of the Markov chain whose transition matrix is {\bf M}. 
\end{theorem}
\begin{proof}
See Appendix~\ref{app:nobacktracking}. 
\end{proof}

With Theorem~\ref{nobacktracking}, we have the following corollary, which states that the vote capacity of a user is proportional to its weighted degree for certain assignments of  the shares of social neighbors and hop-2 social neighbors in the dividing step. 
 \setcounter{theorem}{0} 
\begin{corollary}
\label{specialcase}
When $\alpha=0$, the SBA network is connected,  and for each user $u$, the shares of social neighbors, behavior-sharing social neighbors, and attribute-sharing social neighbors in the dividing step are $w_S = \tau \cdot {d_{u,S}}$, $w_{BS} = \tau \cdot {d_{u,B}}$, and $w_{AS} = \tau \cdot {d_{u,A}}$, respectively, then we have:
\begin{align}
\vec{s}_{vu} = |V_s| \frac{d_u}{D},
\end{align}
where $\tau$ is any positive number, $d_u$ is the weights of all links of $u$ and $D$ is the twice of the total weights of all links in the SBA network, i.e., $D=\sum_{u} d_{u}$.
\end{corollary}
\begin{proof}
See Appendix~\ref{app:specialcase}. 
\end{proof}

\subsection{Importance of Backtracking}
Theorem~\ref{nobacktracking} implies that when there is no backtracking, the converged vote capacity vector is independent with the targeted users. In other words, VIAL with no backtracking predicts the same attribute values for all targeted users. This explains why VIAL with no backtracking achieves suboptimal performance. We will further empirically evaluate the impact of  backtracking strength in our experiments, and we found that VIAL's performance significantly degrades when there is no backtracking. 

\subsection{Time Complexity}
 The major cost of VIAL is from Phase I, which includes computing ${\bf M}$ and iteratively computing the vote capacity vector.  ${\bf M}$ only needs to be computed once and is applied to all targeted users. 
  ${\bf M}$ is a sparse matrix with $O(m)$ non-zero entries, where $m$ is the number of links in the SBA network.  To compute ${\bf M}$,  for every social node, we need to go through its social neighbors and hop-2 social neighbors; and for a hop-2 social neighbor, we need to go through the common attribute/behavior neighbors between the social node and the hop-2 social neighbor. Therefore, the time complexity of computing  ${\bf M}$ is $O(m)$.
 
 Using sparse matrix representation of  ${\bf M}$, the time complexity of each iteration 
(i.e., applying  Equation~\ref{iter}) in computing the vote capacity vector is $O(m)$. Therefore, 
the time complexity 
of computing the vote capacity vector for one targeted user is $O(d\cdot m)$, where $d$ is the 
number of iterations. Thus, the overall time complexity of VIAL is $O(d\cdot m)$ for one targeted user.

   

\section{Data Collection}
We collected a dataset from  Google+ and Google Play to evaluate our VIAL attack and previous attacks. Specifically, we collected social structures and user attributes from Google+, and user review behaviors from Google Play. 
Google assigns each user a  21-digit universal ID, which is used in both Google+ and Google Play. We first collected a social network with user attributes from Google+  via iteratively crawling users' friends. 
Then we crawled review data of users in the Google+ dataset. All the information that we collected is publicly available. 


\subsection{Google+ Dataset} 
Each user 
 in Google+ has  an outgoing friend list (i.e., ``in your
circles''), an incoming friend list (i.e., ``have you in circles''), and a profile.  
Shortly after Google+ was launched in late June 2011, Gong et al.~\cite{Gong12-imc,gong2014joint} 
began to crawl daily snapshots of public
Google+ social network structure and user profiles (e.g., \emph{major}, \emph{employer}, and \emph{cities lived}). 
Their dataset includes 79 snapshots of Google+ collected
 from July 6 to October 11, 2011. Each snapshot was a large Weakly Connected Component
 of Google+ social network at the time of crawling.

We obtained one collected snapshot from Gong et al.~\cite{Gong12-imc,gong2014joint}. 
To better approximate friendships between users, we construct an undirected social network from the crawled Google+ dataset via keeping an undirected link between a user $u$ and $v$ if $u$ is in $v$'s both incoming friend list and outgoing friend list. 
After preprocessing,  our Google+ dataset consists of 1,111,905 users and  5,328,308 undirected social links. 



\myparatight{User attributes} We consider three attributes, \emph{major}, \emph{employer}, and \emph{cities lived}. We note that, although we focus on these attributes that are available to us at a large scale, our attack is also applicable to infer other attributes such as sexual orientation, political views, and religious views. 
 Moreover, some targeted users might not view inferring these attributes as an privacy attack, but 
 an attacker can leverage these attributes to further link users across online social networks~\cite{Bartunov12,goga2013large,stylometrylink,userlinkAcrossSitesWWW13} or even link them with offline records to perform more serious security and privacy attacks~\cite{kanonymity,offlineattack}. 
  
 We take the strings input by a user in its Google+ profile as attribute values.
We found that  most attribute values are owned by a small number of users while some are owned by a large number of users. 
Users could fill in their profiles freely in Google+, which could be  one reason that we observe many infrequent attribute values. Specifically, different users might have different names for the same attribute value. For instance, the major of Computer Science could also be abbreviated as CS by some users. Indeed, we find that 20,861 users have Computer Science as their major and 556 users have CS as their major in our dataset. Moreover, small typos (e.g., one letter is incorrect) in the free-form inputs make the same attribute value be treated as different ones.  Therefore, we manually label a set of attribute values.

\emph{1) Major.}
We consider the top-100 majors that are claimed by the most users. We manually merge the majors that actually refer to the same one, e.g., Computer Science and CS, Btech and Biotechnology. After preprocessing, we obtain 62 distinct majors. 8.4\% of users in our dataset have at least one of these majors.

%

\emph{2) Employer.} Similar to major, we select the top-100 employers that are claimed by the most users and manually merge the employers that  refer to the same one.  We obtain 78 distinct employers, and 3.1\% of users have at least one of these employers.

%

\emph{3) Cities lived.} Again, we select the top-100 cities in which most users in the Google+ dataset claimed they have lived in. After we  manually merge the cities that actually refer to the same one, we obtain 70 distinct cities. 8\% of users have at least one of these attribute values.  


\myparatight{Summary and limitations} In total, we consider 210 popular distinct attribute values, including 62 majors, 78 employers, and 70 cities. We acknowledge that our Google+ dataset might not be a representative sample of the recent entire Google+ social network, and thus the inference attack success rates obtained in our experiments might not represent those of the entire Google+ social network. 


\subsection{Crawling Google Play} 
There are 7 categories of items in Google Play. They are \emph{apps, tv, movies, music, books, newsstand,} and \emph{devices}. 
Google Play provides two review mechanisms for users to provide feedback on an item. They are  the \emph{liking} mechanism and the \emph{rating} mechanism.  In the liking mechanism, a user simply clicks a like button to express his preference about an item. In the rating mechanism, a user gives a rating score which is in the set \{1,2,3,4,5\} as well as a detailed comment to support his/her rating. A score of 1 represents low preference and a score of 5 represents high preference. We call a user \emph{reviewed} an item if the user rated or liked the item.  

User reviews are publicly available in Google Play. Specifically, after a user $u$ logs in Google Play, $u$ can view the list of items reviewed by any user $v$ once $u$ can obtain $v$'s Google ID. We crawled the list of items reviewed by each user in the Google+ dataset. 

We find that 33\% of users in the Google+ dataset have reviewed at least one
item. In total, we collected 260,245 items and 3,954,822 reviews.   
Since items with too few reviews might not be informative to distinguish users with different attribute values, 
we use items that were reviewed by at least 5 users. After preprocessing, we have 48,706 items and 3,635,231 reviews.

\begin{table}[t]\renewcommand{\arraystretch}{1}
\centering
\caption{Basic statistics of our SBA.}
\addtolength{\tabcolsep}{-4pt}
\begin{tabular}{|c|c|c|c|c|c|} \hline
\multicolumn{3}{|c|}{\small \#nodes} & \multicolumn{3}{c|}{\small \#links} \\ \hline
{{\small social}} & {{\small behavior}} &{{\small attri.}} &{{\small social}} &{{\small behavior}} &{{\small attri.}}\\ \hline
 {\small 1,111,905}  & {\small 48,706} & {\small 210} & {\small 5,328,308} &{\small  3,635,231} & {\small 269,997} \\ \hline
\end{tabular}
\label{data}
\end{table} 

\subsection{Constructing SBA Networks}
We take each user in the Google+ dataset as a social node and links between them as social links.  
For each item in our Google Play dataset, we add a corresponding behavior node. 
 If a user reviewed an item, we create a link between the corresponding social node and the corresponding behavior node. That a user reviewed an item means that the user once used the item. Using similar items could indicate similar  interests, user characteristics, and user attributes. To predict attribute values, we further add  additional attribute nodes to represent attribute values, and
we create a link between a social node and an attribute node if the user has the attribute value. 
Table~\ref{data} shows the basic statistics of our constructed SBA for predicting attribute values. 

In this work, we set the weights of all links in the SBA to be 1. 
Therefore, our attacking result represents a lower bound on what an attacker can achieve in practice. 
An attacker could leverage machine learning techniques (we discuss one in Section~\ref{app:weight}) to
learn link weights to further improve success rates. 
  
%


\section{Experiments}
\subsection{Experimental Setup}
\label{sec:setup}
We describe the metrics we adopt to evaluate various attacks, training and testing, and parameter settings.

\myparatight{Evaluation metrics} All attacks we evaluate essentially assign a score for each candidate attribute value. Given a targeted user $v$,  we predict top-$K$ candidate attribute values that have the highest scores for each attribute including major, employer, and cities lived. We use Precision, Recall, and F-score to evaluate the  top-$K$ predictions. In particular, Precision is the fraction of predicted attribute values that belong to $v$. 
Recall is the fraction of $v$'s attribute values that are among the predicted $K$ attribute values. We address score ties in the manner described by McSherry and Najork~\cite{McSherry08}. 
Precision characterizes how accurate an attacker's inferences are while Recall characterizes how many user attributes are corrected inferred by an attacker. 
In particular, Precision for top-$1$ prediction is the fraction of users that the attacker can correctly infer at least one attribute value. 
F-score is the harmonic mean of Precision and Recall, i.e., we have
$$\text{F-score} = \frac{2 \cdot \text{Precision} \cdot \text{Recall}}{\text{Precision} + \text{Recall}}.$$ 
Moreover, we average the three metrics over all targeted users. 
For convenience, we will also use P, R, and F to represent Precision, Recall, and F-Score, respectively.  

We also define \emph{performance gain} and \emph{relative performance gain} of one attack $\mathscr{A}$ over another attack $\mathscr{B}$ to compare their relative performances. We take Precision as an example to show their definitions as follows:
\begin{align}
&\text{\bf Performance gain:} \nonumber \\
\Delta \text{P} &=  \text{Precision}_{\mathscr{A}} - \text{Precision}_{\mathscr{B}} \nonumber 
\end{align}
\begin{align}
&\text{\bf Relative performance gain:} \nonumber \\
\Delta \text{P}\% &=  \frac{\text{Precision}_{\mathscr{A}} - \text{Precision}_{\mathscr{B}}}{\text{Precision}_{\mathscr{B}}}\times 100\% \nonumber
\end{align}

\myparatight{Training and testing} For each attribute value, we sample 5 users uniformly at random from the users that have the attribute value and have reviewed at least 5 items, and we treat them as test (i.e., targeted) users. In total, we have around 1,050 test users. 
For test users, we remove their attribute links from the SBA network and use them as groundtruth. We repeat the experiments 10 times and average the evaluation metrics over the 10 trials. 


\myparatight{Parameter settings}  In the dividing step, we set equal shares for social neighbors, behavior-sharing social neighbors, and attribute-sharing social neighbors, i.e.,  $w_S=w_{BS}=w_{AS}=\frac{1}{3}$. The number of iterations to compute the vote capacity vector is $d=\lfloor \text{log}\ |V_s| \rfloor$=20, after which the vote capacity vector converges. 
Unless otherwise stated, we set the backtracking strength $\alpha=0.1$. 

\subsection{Compared Attacks}
We compare VIAL with friend-based attacks, behavior-based attacks, and attacks that use both social structures and behaviors. These attacks essentially assign a score for each candidate attribute value, and return the $K$ attribute values that have the highest scores. Suppose $v$ is a test user and $a$ is an attribute value, and we denote by $S(v,a)$ the score assigned to $a$ for $v$.

\myparatight{Random}  This baseline method computes the fraction of users in the training dataset that have a certain attribute value $a$,  and it treats such fraction as the score $S(v,a)$ for all test users.

\begin{figure*}[!t]
\centering
\subfloat[Precision]{\includegraphics[width=0.33 \textwidth]{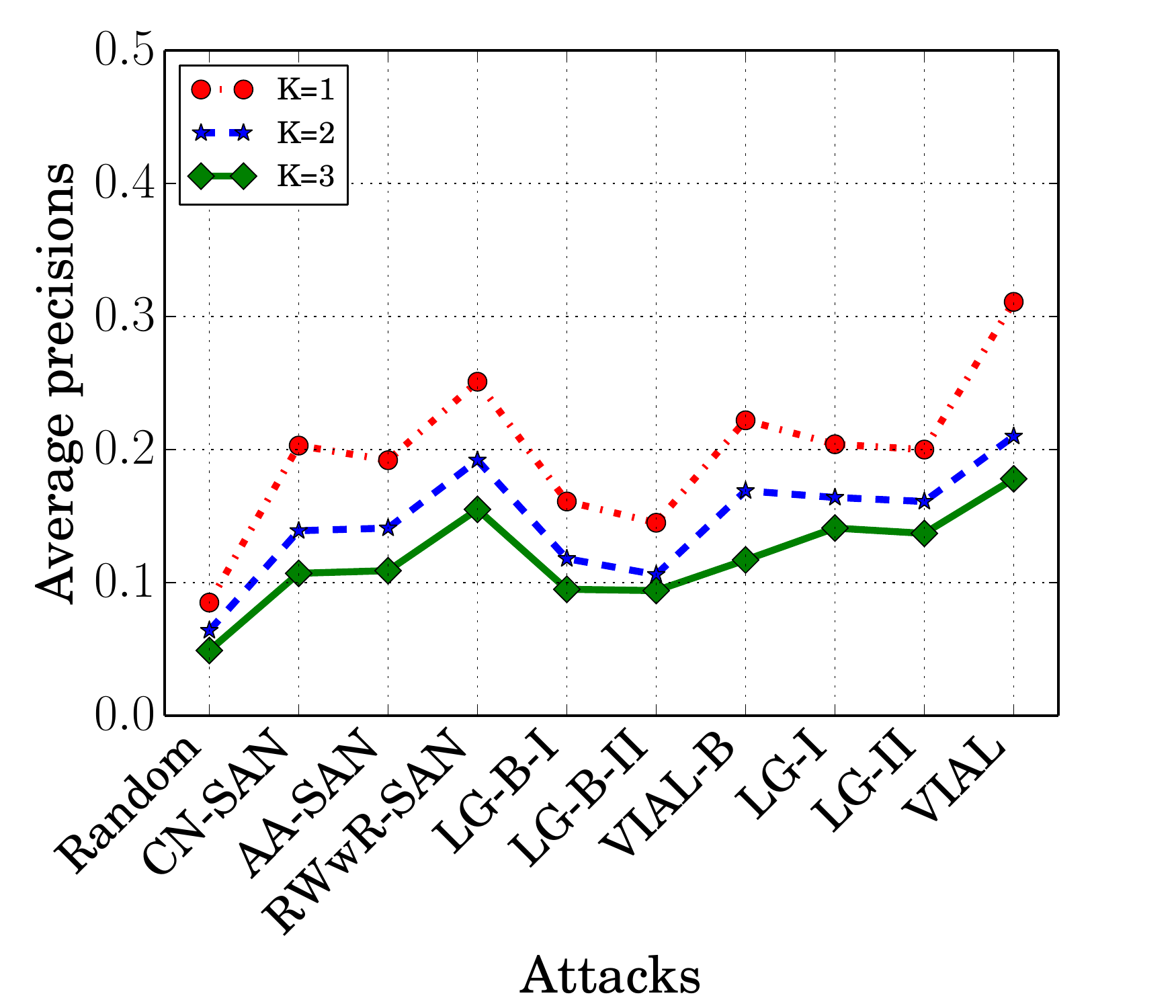}} 
\subfloat[Recall]{\includegraphics[width=0.33 \textwidth]{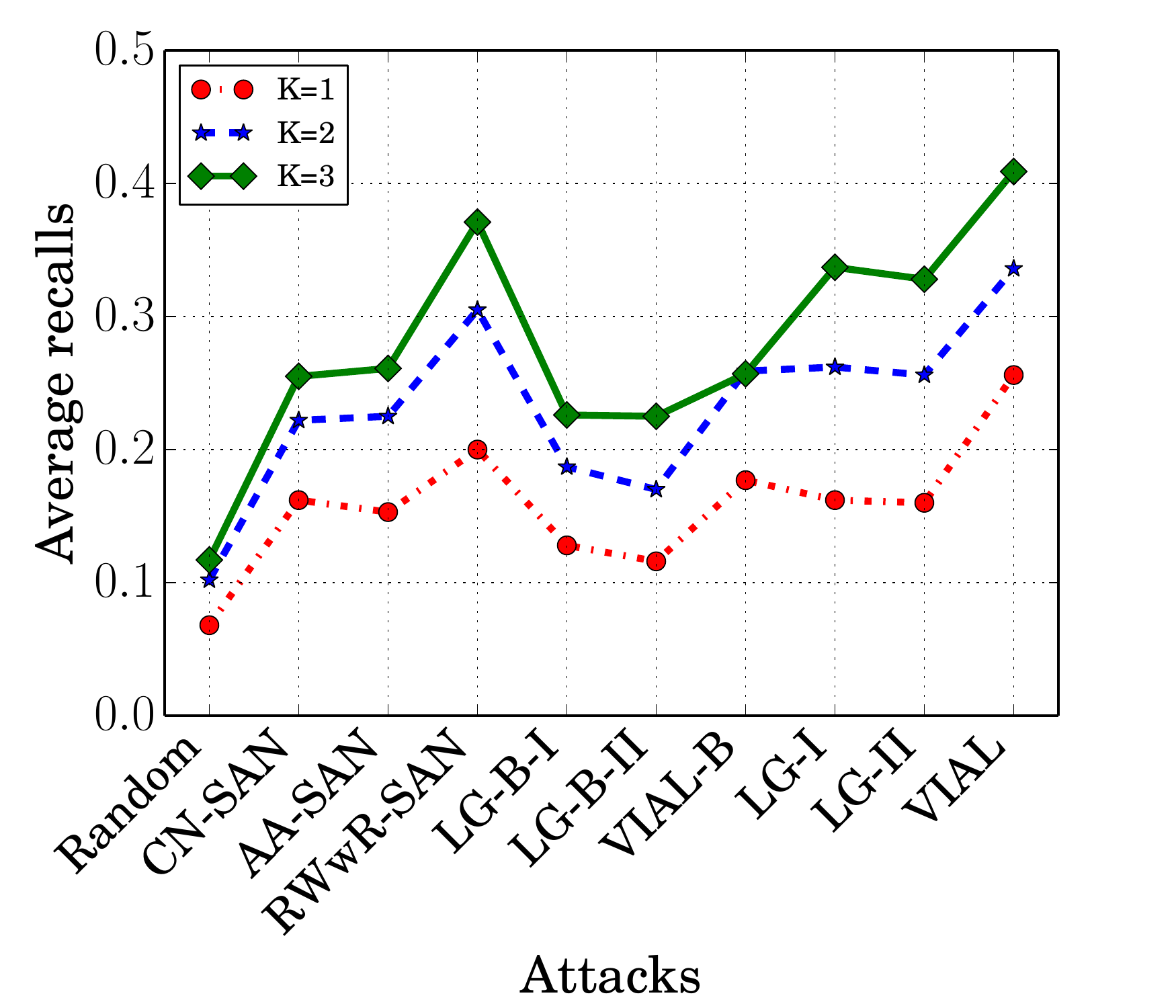}} 
\subfloat[F-score]{\includegraphics[width=0.33 \textwidth]{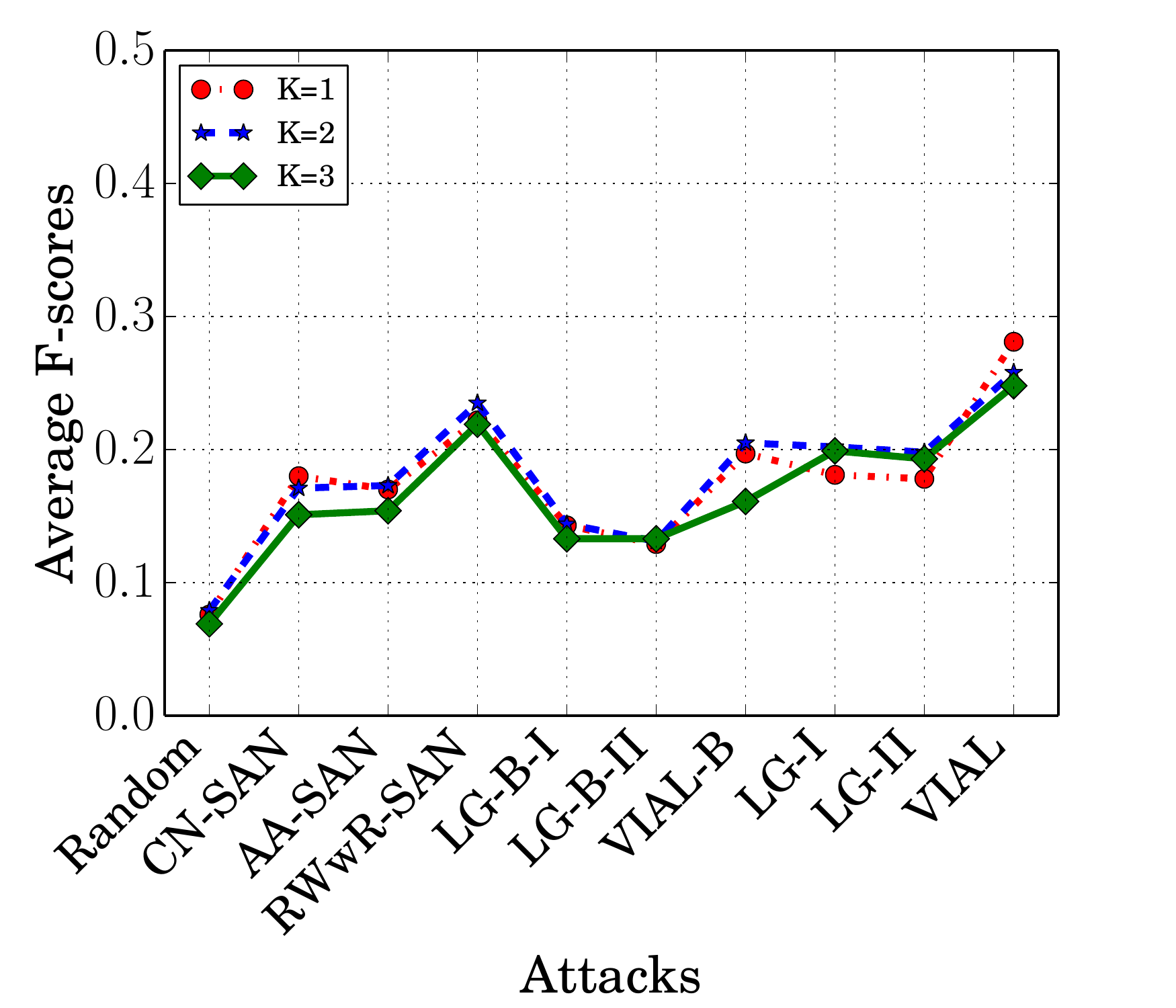}} 

\caption{Precision, Recall, and F-Score for inferring majors. Although these attacks do not have temporal orderings, we connect them via curves in the figures to better contrast them.}
\vspace{4mm}
\label{major-infer}
\end{figure*}

%
%

\begin{figure*}[!t]
\vspace{4mm}
\centering
\subfloat[Precision]{\includegraphics[width=0.33 \textwidth]{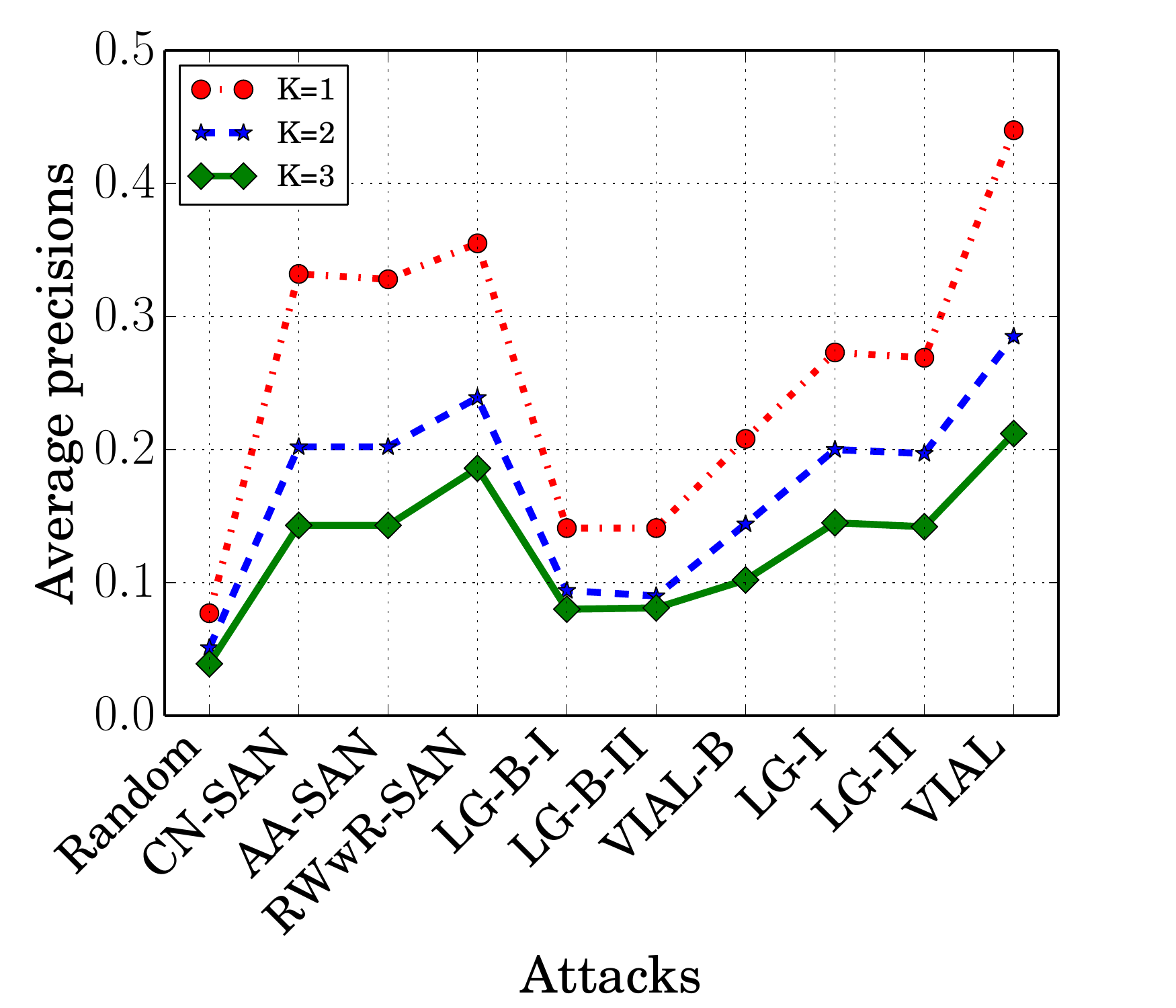}} 
\subfloat[Recall]{\includegraphics[width=0.33 \textwidth]{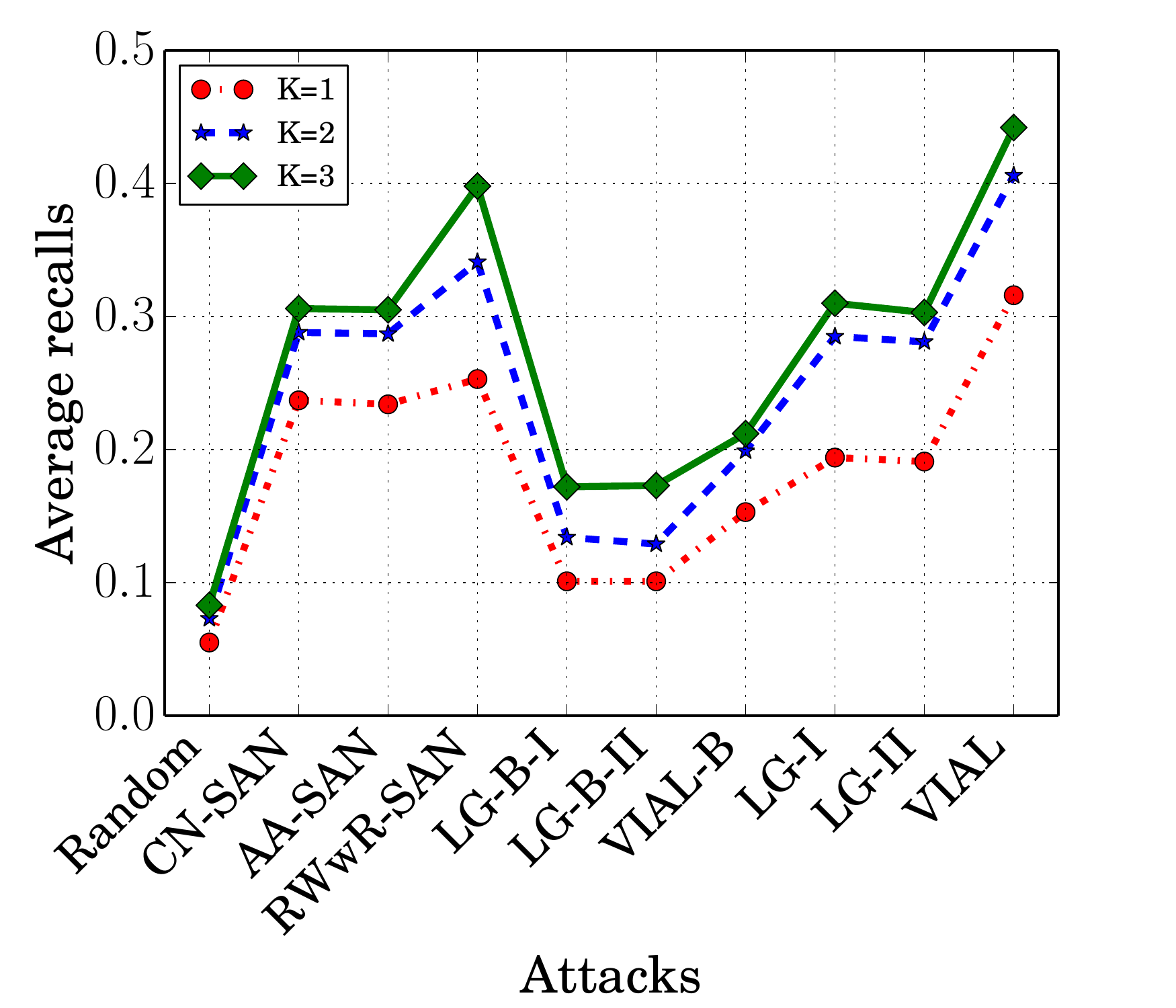}} 
\subfloat[F-score]{\includegraphics[width=0.33 \textwidth]{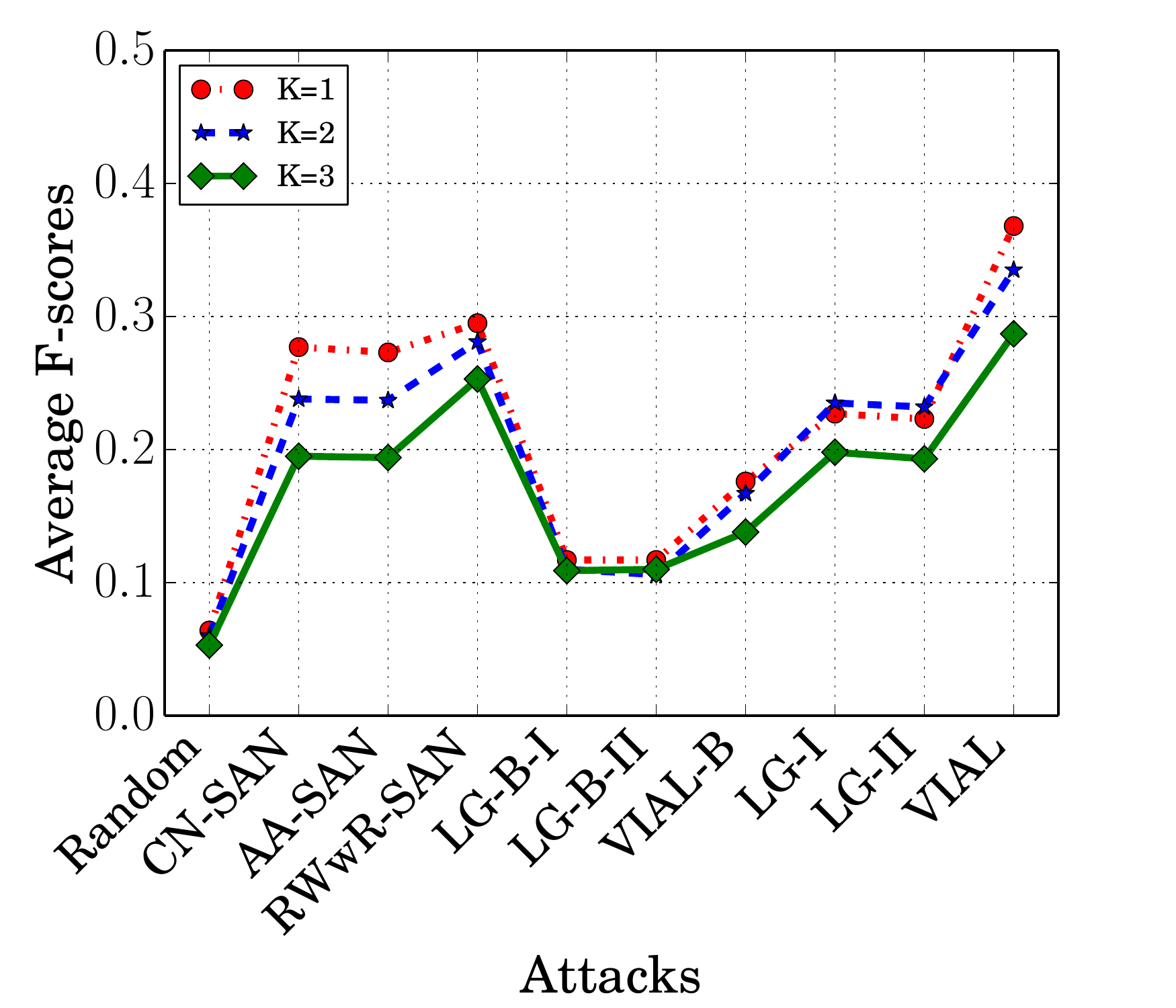}} 

\caption{Precision, Recall, and F-Score for inferring employers.}
\vspace{4mm}
\label{employer-infer}
\end{figure*}

%

\begin{figure*}[!t]
\vspace{4mm}
\centering
\subfloat[Precision]{\includegraphics[width=0.33 \textwidth]{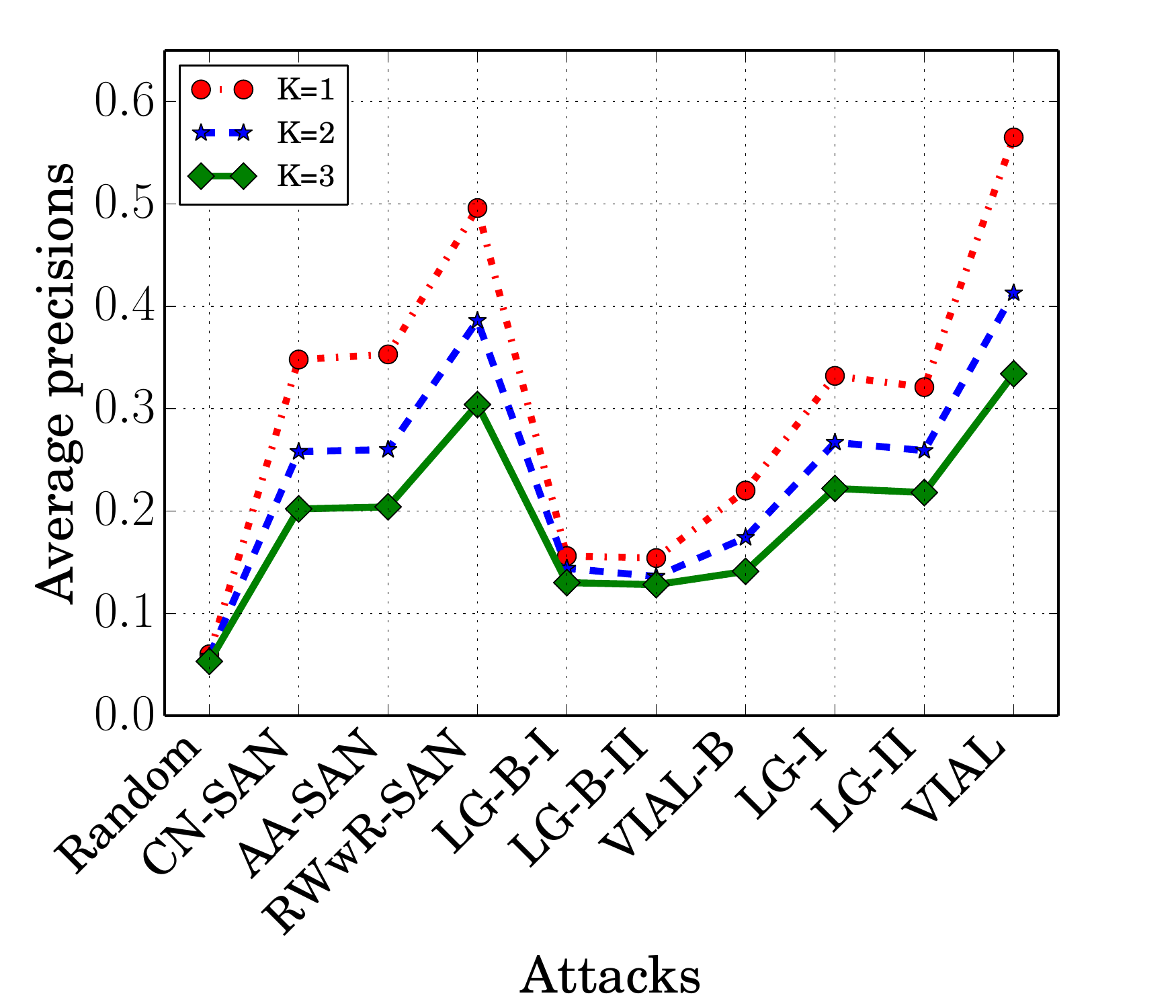}} 
\subfloat[Recall]{\includegraphics[width=0.33 \textwidth]{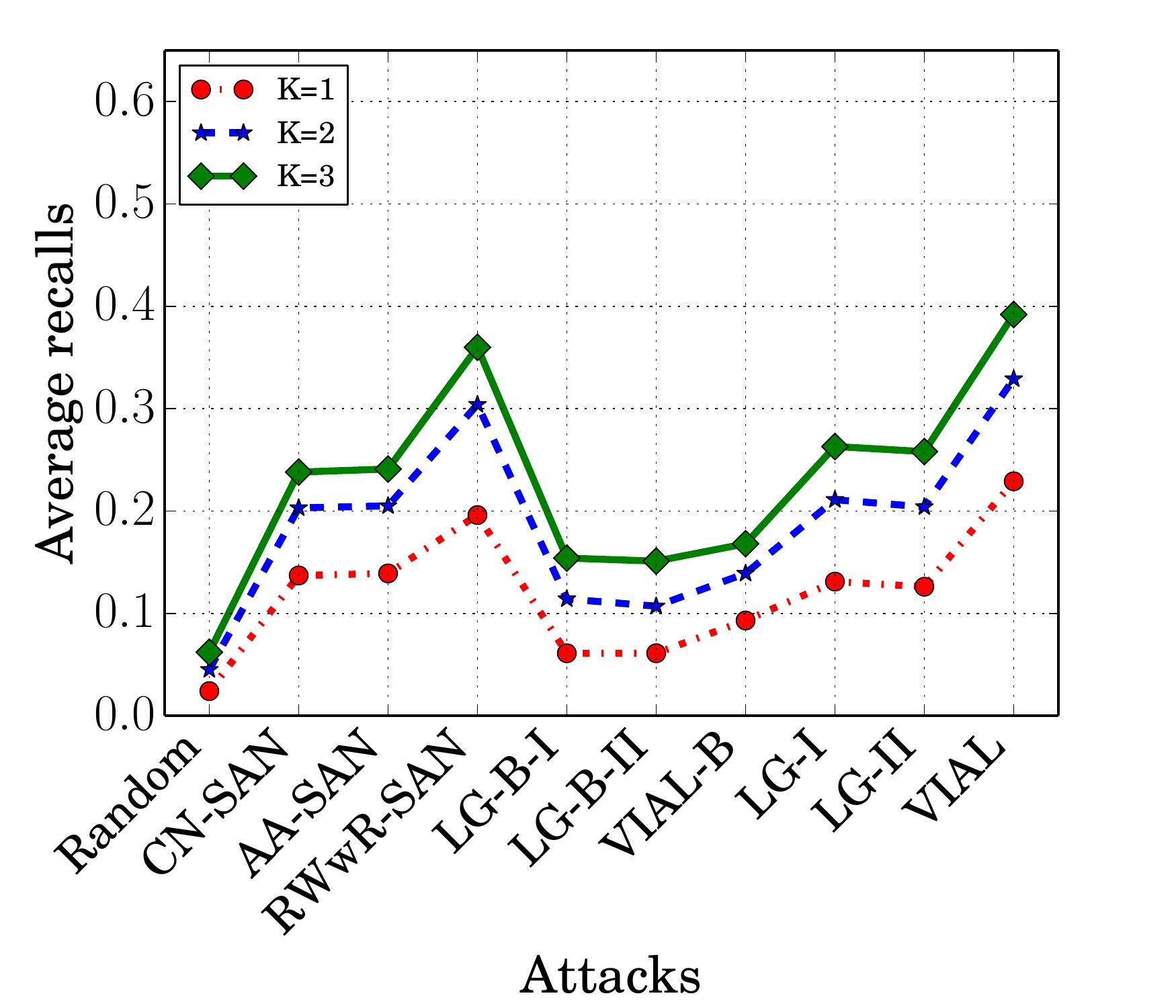}} 
\subfloat[F-score]{\includegraphics[width=0.33 \textwidth]{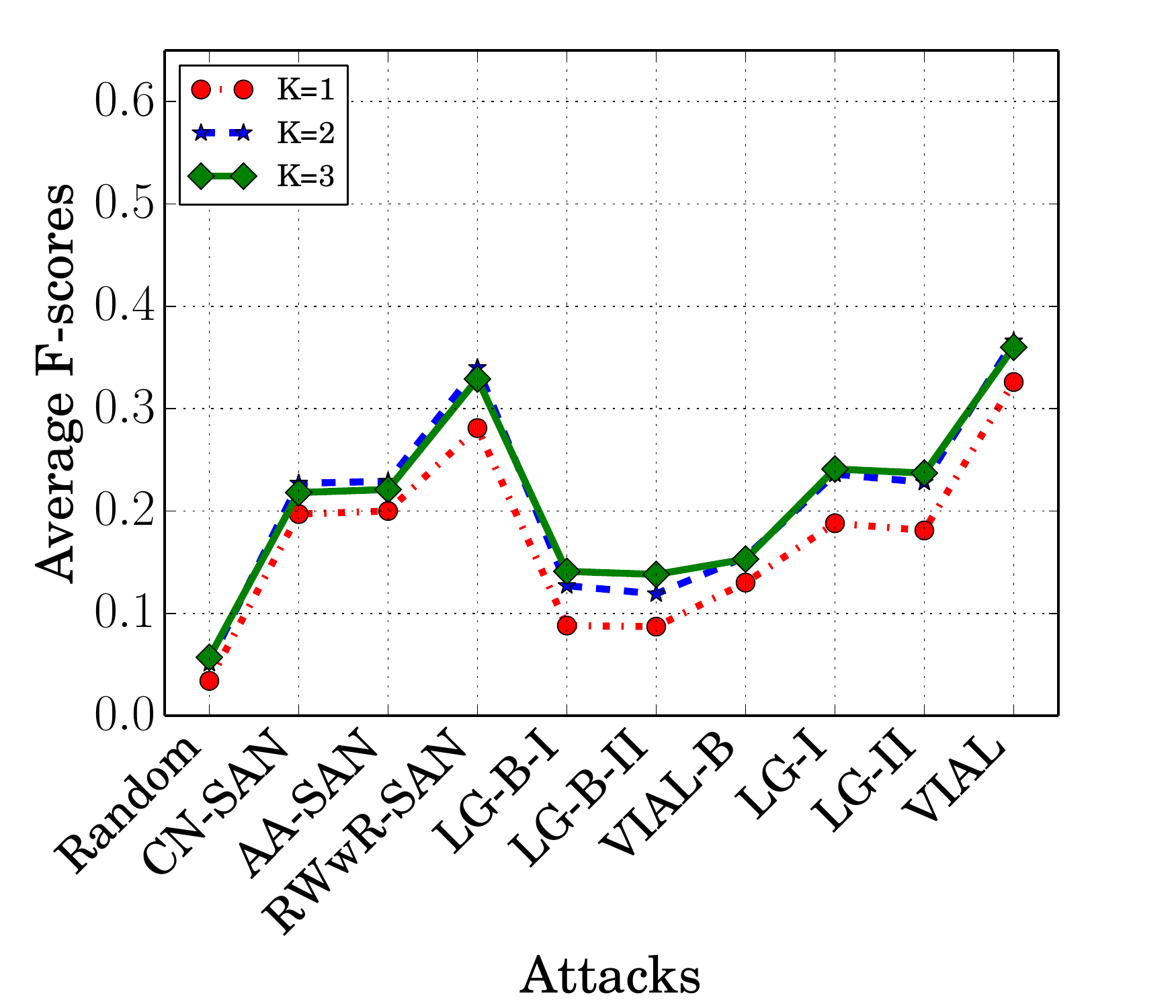}} 

\caption{Precision, Recall, and F-Score for inferring cities.}
\label{city-infer}
\end{figure*}

%
%

\myparatight{Friend-based attacks} We compare with three friend-based attacks, i.e., CN-SAN, AA-SAN, and RWwR-SAN~\cite{gong2014joint}. They were shown to outperform previous attacks such as LINK~\cite{Zheleva09,gong2014joint}. 
\begin{itemize}
\item {\bf CN-SAN.} $S(v,a)$ is the number of  common social neighbors between $v$ and $a$.
\item {\bf AA-SAN.} This attack weights the importance of each common social neighbor between $v$ and $a$ proportional to the inverse of the log of its number of neighbors. Formally, $S(v,a)=\sum_{u\in \Gamma_{v, S} \cap \Gamma_{a, S}} \frac{1}{\text{log}|\Gamma_{u}|}$.

\item {\bf RWwR-SAN.} RWwR-SAN augments the social network with additional attribute nodes. Then it performs a random walk that is initialized from the test user $v$ on the augmented graph. The stationary probability of the attribute node that corresponds to $a$ is treated as the score $S(v,a)$. 

\end{itemize}

\begin{table}[t]\renewcommand{\arraystretch}{1}
\centering
\caption{Performance gains and relative performance gains of RWwR-SAN over other friend-based attacks, where $K=1$.  Results are averaged over all attributes. We find that RWwR-SAN is the best friend-based attack.}
{
\begin{tabular}{|c|c|c|c|c|c|c|} \hline
{\small Attack}&{\small $\Delta$P} &{\small $\Delta$P\%} & {\small $\Delta$R}& {\small $\Delta$R\%}&{\small $\Delta$F}&{\small $\Delta$F\%}\\ \hline
{\small CN-SAN}&{\small 0.07}&{\small 24\%}&{\small 0.04}&{\small 24\%}&{\small 0.05}&{\small 24\%}\\ \hline
{\small AA-SAN}&{\small 0.08}&{\small 26\%}&{\small 0.04}&{\small 26\%}&{\small 0.05}&{\small 26\%}\\ \hline
\end{tabular}}
\label{friend-based}
\end{table}

\begin{table}[t]\renewcommand{\arraystretch}{1}
\centering
\caption{Performance gains and relative performance gains of VIAL-B over other behavior-based attacks, where $K=1$. We find that VIAL-B is the best behavior-based attack.}
{
\begin{tabular}{|c|c|c|c|c|c|c|} \hline
{\small Attack}&{\small $\Delta$P} &{\small $\Delta$P\%} & {\small $\Delta$R}& {\small $\Delta$R\%}&{\small $\Delta$F}&{\small $\Delta$F\%}\\ \hline
{\small LG-B-I}&{\small 0.06}&{\small 42\%}&{\small 0.04}&{\small 47\%}&{\small 0.05}&{\small 45\%}\\ \hline
{\small LG-B-II}&{\small 0.07}&{\small 47\%}&{\small 0.05}&{\small 52\%}&{\small 0.06}&{\small 50\%}\\ \hline
\end{tabular}}
\label{behavior-based}
\end{table} 

\myparatight{Behavior-based attacks} We also evaluate three behavior-based attacks. 
\begin{itemize}
\item  {\bf Logistic regression (LG-B-I)~\cite{weinsberg2012blurme}.} LG-B-I treats each attribute value  as a class and learns a multi-class logistic regression classifier with the training dataset. Specifically,  LG-B-I extracts a feature vector whose length is the number of items for each user that has review data, and a feature has a value of the rating score that the user gave to the corresponding item. Google Play allows users to rate or like an item, and we treat a liking as a rating score of 5.  For a test user, the learned logistic regression classifier returns a posterior probability distribution over the possible attribute values, which are used as the scores  $S(v,a)$. Weinsberg et al.~\cite{weinsberg2012blurme} showed that logistic regression classifier outperforms other classifiers including SVM~\cite{Cortes95} and Naive Bayes~\cite{McCallum98}. 

\item {\bf Logistic regression with binary features (LG-B-II)~\cite{kosinski2013private}.} The difference between LG-B-II and LG-B-I is that LG-B-II extracts binary feature vectors for users. Specifically, a feature has a value of 1 if the user has reviewed the corresponding item.   

\item {\bf VIAL-B.} A variant of VIAL that only uses behavior data. Specifically, we remove social links from the SBA network and perform our VIAL attack using the remaining links.  
\end{itemize}

\myparatight{Attacks combining social structures and behaviors}
Intuitively, we can combine social structures and behaviors via concatenating social structure features with behavior features. 
We  compare with two such attacks.
 \begin{itemize}
\item {\bf Logistic regression (LG-I).}  LG-I extracts a binary feature vector whose length is the number of users from social structures for each user, and a feature has a value of 1 if the user is a friend of the person that corresponds to the feature. Then LG-I  concatenates this feature vector with the one used in LG-B-I and learns multi-class logistic regression classifiers. 

\item {\bf Logistic regression with binary features (LG-II).} LG-II concatenates the binary social structure feature vector with the binary behavior feature vector used by LG-B-II. 

\end{itemize}

We use the popular package LIBLINEAR~\cite{liblinear08} to learn logistic regression classifiers. 

\begin{table}[t]\renewcommand{\arraystretch}{1}
\centering
\caption{Performance gains and relative performance gains of VIAL over other attacks combining social structures and behaviors, where $K=1$.  We find that VIAL substantially outperforms other attacks.}
{
\begin{tabular}{|c|c|c|c|c|c|c|} \hline
{\small Attack}&{\small $\Delta$P} &{\small $\Delta$P\%} & {\small $\Delta$R}& {\small $\Delta$R\%}&{\small $\Delta$F}&{\small $\Delta$F\%}\\ \hline
{\small LG-I}&{\small 0.17}&{\small 61\%}&{\small 0.10}&{\small 65\%}&{\small 0.13}&{\small 63\%}\\ \hline
{\small LG-II}&{\small 0.18}&{\small 65\%}&{\small 0.11}&{\small 69\%}&{\small 0.13}&{\small 67\%}\\ \hline
\end{tabular}}
\label{both-based}
\end{table} 

%
%

\begin{table}[t]\renewcommand{\arraystretch}{1}
\centering
\caption{Performance gains and relative performance gains of VIAL over Random, RWwR-SAN (the best friend-based attack),  and VIAL-B (the best behavior-based attack), where $K=1$. }
{
\addtolength{\tabcolsep}{-3pt}
\begin{tabular}{|c|c|c|c|c|c|c|} \hline
{\small Attack}&{\small $\Delta$P} &{\small $\Delta$P\%} & {\small $\Delta$R}& {\small $\Delta$R\%}&{\small $\Delta$F}&{\small $\Delta$F\%}\\ \hline
{\small Random}&{\small 0.36}&{\small 526\%}&{\small 0.22}&{\small 535\%}&{\small 0.27}&{\small 534\%}\\ \hline
{\small RWwR-SAN}&{\small 0.07}&{\small 20\%}&{\small 0.05}&{\small 23\%}&{\small 0.06}&{\small 22\%}\\ \hline
{\small VIAL-B}&{\small 0.22}&{\small 102\%}&{\small 0.13}&{\small 99\%}&{\small 0.16}&{\small 100\%}\\ \hline
\end{tabular}}
\label{gain}
\end{table} 

\subsection{Results}
Fig.~\ref{major-infer}-Fig.~\ref{city-infer} demonstrate the Precision, Recall, and F-score for top-$K$  inference of major, employer, and city, where $K=1,2,3$. Table~\ref{friend-based}-Table~\ref{gain} compare different attacks using results that are averaged over all attributes.  Our metrics are averaged  over 10 trials.  We find that standard deviations of the metrics are very small, and thus we do not show them for simplicity. 
Next, we describe several key observations we have made from these results.

\myparatight{Comparing friend-based attacks} We find that RWwR-SAN  performs the best among the friend-based attacks. Our observation is consistent with the previous work~\cite{gong2014joint}. To better illustrate the difference between the friend-based attacks, we show the performance gains and relative performance gains of RWwR-SAN over other friend-based attacks in Table~\ref{friend-based}. Please refer to Section~\ref{sec:setup} for formal definitions of (relative) performance gains. 
The (relative) performance gains are averaged over all attributes (i.e., major, employer, and city). The reason why RWwR-SAN outperforms other friend-based attacks is that RWwR-SAN performs a random walk among the augmented graph, which better leverages the graph structure, while other attacks simply count the number of common neighbors or weighted common neighbors.  

\myparatight{Comparing behavior-based attacks} We find that VIAL-B  performs the best among the behavior-based attacks.  Table~\ref{behavior-based}  shows the average performance gains and relative performance gains of VIAL-B over other behavior-based attacks.  Our results indicate that our graph-based attack is a better way to leverage behavior structures, compared to LG-B-I and LG-B-II, which flatten the behavior structures into feature vectors. Moreover, LG-B-I and LG-B-II achieve very close performances, which indicates that the rating scores carry little information about user attributes. 


\myparatight{Comparing attacks combining social structure and behavior} We find that VIAL performs the best among the attacks combining social structures and behaviors. Table~\ref{both-based} shows the average performance gains and relative performance gains of VIAL over other attacks. Our results  imply that, compared to flattening the structures into feature vectors, our graph-based attack can better integrate social structures and user behaviors. 


\myparatight{Comparing VIAL with the best friend-based attack and the best behavior-based attack} Table~\ref{both-based} shows the average performance gains and relative performance gains of VIAL over Random, the best friend-based attack, and the best behavior-based attack. We find that VIAL significantly outperforms these attacks, indicating the importance of combining social structures and behaviors to perform attribute inference. This implies that, when an attacker wants to attack user privacy via inferring their private attributes, the attacker can successfully attack substantially more users using VIAL.


 \begin{figure}[!t]
\centering
{\includegraphics[width=0.45 \textwidth]{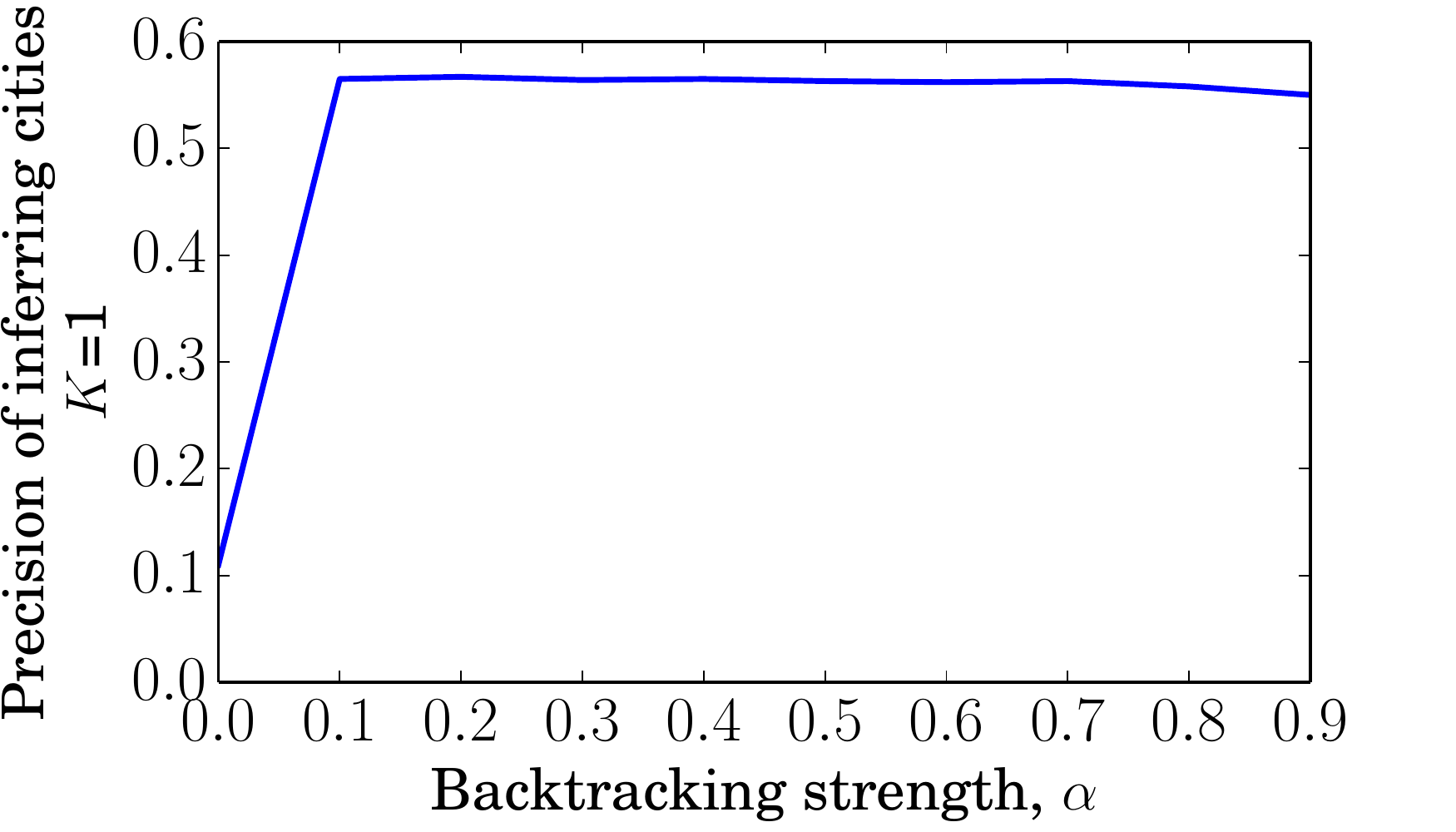}}
\caption{Impact of the backtracking strength  on the Precision of VIAL for  inferring cities. We observe that backtracking substantially improves VIAL's performance.}
 \label{impact-of-alpha}
\end{figure}

\myparatight{Impact of backtracking strength} Fig.~\ref{impact-of-alpha} shows the impact of backtracking strength on the Precision of VIAL for  inferring cities. According to Theorem~\ref{converge}, VIAL with $\alpha=1$ reduces to random guessing, and thus we do not show the corresponding result in the figure. $\alpha=0$ corresponds to the case in which VIAL does not use backtracking. We observe that not using backtracking substantially decreases the performance of VIAL. The reason might be that 1) $\alpha=0$ makes VIAL predict the same attribute values for all test users, according to Theorem~\ref{nobacktracking}, and 2) a user' attributes are close to the user in the SBA network and backtracking makes it  more likely for votes to be distributed among these attribute nodes. Moreover, we find that inference accuracies are  stable across  different backtracking strengths once they are larger than 0. The reason is that when we increase the backtracking strength,  attribute values receive different votes, but the ones with top ranked votes only change slightly. We observe similar results for  other attributes.


 \begin{figure}[!t]
\centering
{\includegraphics[width=0.45 \textwidth]{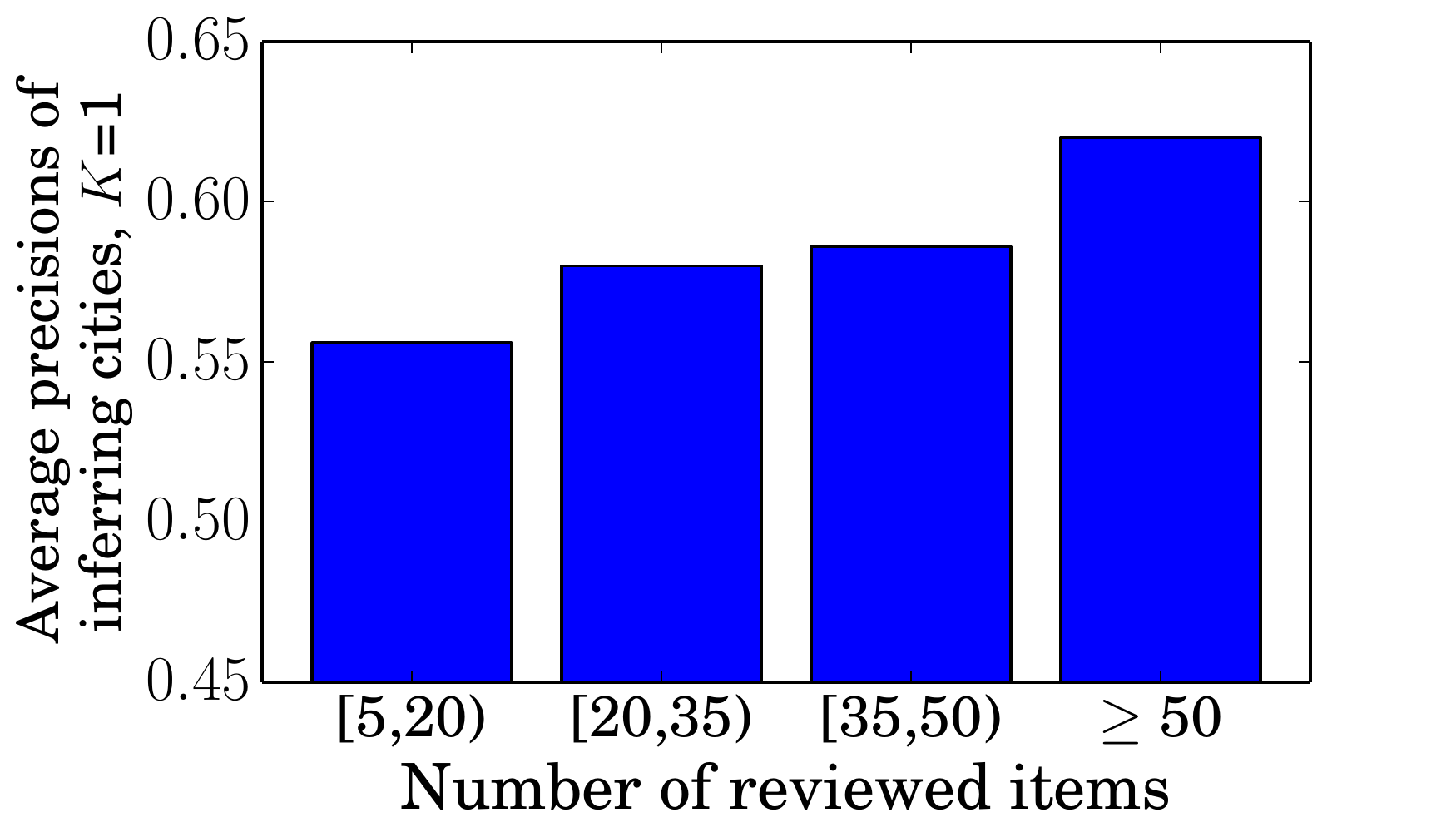}}
\caption{ Impact of the number of reviewed items  on the Precision of our attack VIAL for  inferring cities. We observe that, when users share more behaviors, our attack is able to more accurately predict their attributes.}
 \label{impact-of-behavior}
\end{figure}

\myparatight{Impact of the number of reviewed items} Figure~\ref{impact-of-behavior} shows the Precision as a function of the number of reviewed items for inferring cities lived. We average  Precisions for test users whose number of reviewed items falls under a certain interval (i.e., [5,20), [20,35), [35,50), or $\geq$ 50).  We observe that our attack  can more accurately infer attributes for users who share more digital behaviors (i.e., reviewed items in our case). 

 \begin{figure}[!t]
\centering
\includegraphics[width=0.45 \textwidth]{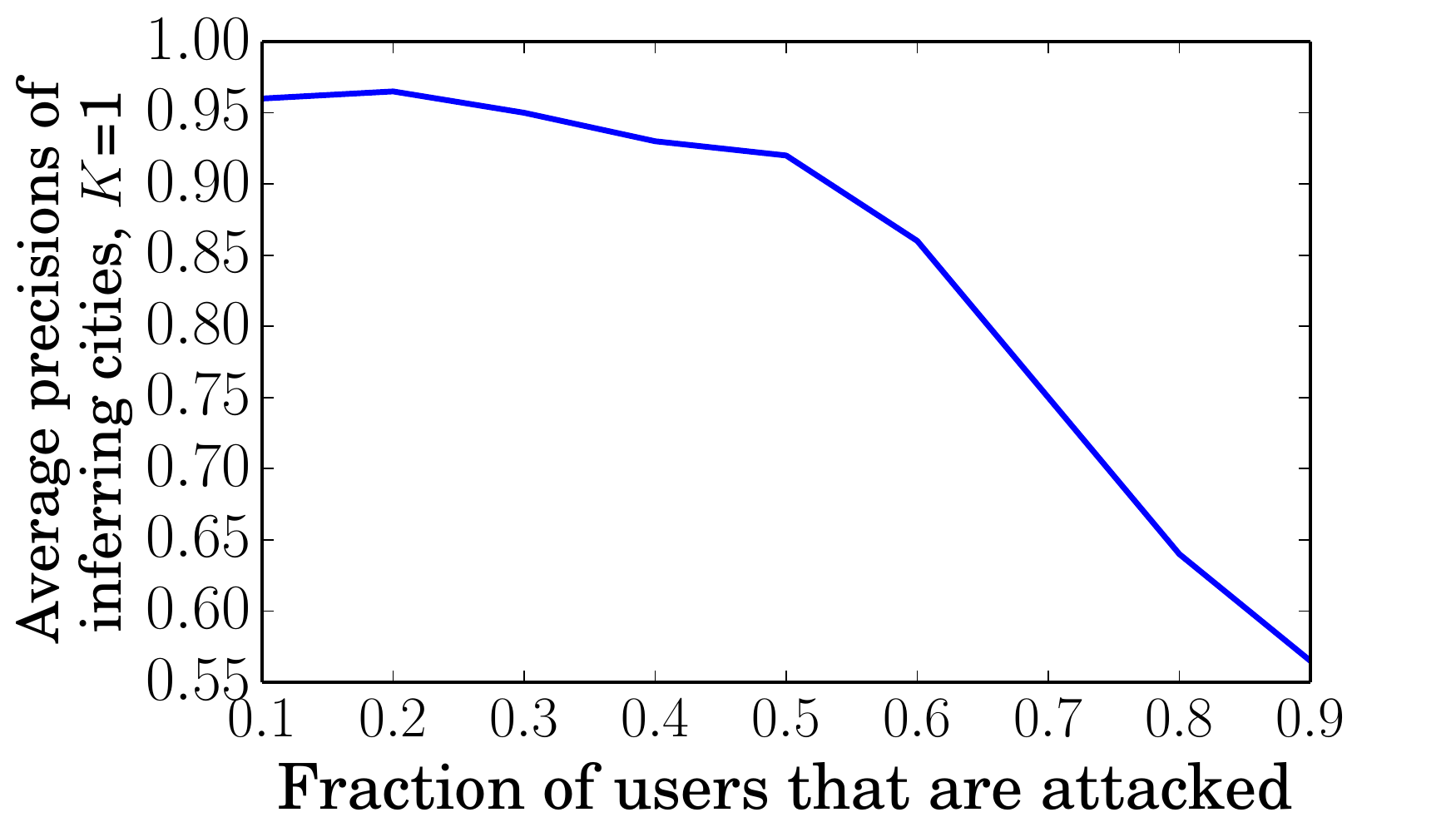} 
\caption{Confidence estimation: trade-off between the Precision of our attack  and the fraction of test users that are attacked. An attacker can substantially improve the attack success rates when attacking less users that are selected by our confidence estimator. }
\label{impact-of-confidence}
\end{figure}


\myparatight{Confidence estimation} 
Figure~\ref{impact-of-confidence} shows the trade-off between the Precision and the fraction of users that are attacked via our confidence estimator. 
We observe that an attacker can increase the Precision ($K=1$) of inferring cities from 0.57 to  over 0.92 
if the attacker attacks a half of the test users that are selected via confidence estimation. 
We also tried the confidence estimator  called \emph{gap statistic}~\cite{writingstyle}, in which the confidence score for a targeted user is the difference between the score of the highest ranked attribute value and the score of the second highest ranked one. Our confidence estimator slightly outperforms gap statistic because a test user could have multiple attribute values and our attack could produce close scores for them. 


\section{Discussion}
\label{app:weight}

This work focuses on propagating vote capacity among the SBA network with given link weights, and our method VIAL is applicable to any link weights. However, it is an interesting future work to learn the link weights, which could further improve the attacker's success rates. 
In the following, we discuss one possible approach to learn link weights. 
Phase I of VIAL essentially iteratively computes the vote capacity vector according to Equation~\ref{iter}. Therefore, Phase 1 of VIAL  can be viewed as performing a \emph{random walk with a restart}~\cite{Tong06} on the subgraph consisting of social nodes and social links, where the matrix ${\bf M}^T$ and $\alpha$ are the transition matrix and restart probability of the random walk, respectively.
Therefore, the attacker could adapt \emph{supervised random walk}~\cite{supervisedRandomWalk} to learn the link weights. Specifically, the attacker already has a set of users with publicly available attributes and the attacker can use them as a training dataset to learn the link weights;  
the attacker removes these attributes from the SBA network as ground truth, and the link weights are learnt such that VIAL can predict attributes for these users the most accurately.

\section{Related Work}


\myparatight{Friend-based attribute inference} He et al.~\cite{He06} transformed attribute inference to Bayesian inference  on a Bayesian network that is constructed using the social links between users.
They evaluated their method using a  LiveJournal social network dataset with \emph{synthesized} user attributes. Moreover,  it is well known in the machine learning community that Bayesian inference is not scalable. Lindamood et al.~\cite{Lindamood09} modified Naive Bayes classifier to incorporate social links and other attributes of users  to infer some attribute. For instance, to infer a user's major, their method used the user's other attributes such as employer and cities lived, the user's social friends and their attributes. 
However,  their approach is not applicable to users that share no attributes at all. 

Zheleva and Getoor~\cite{Zheleva09} studied various approaches to consider both social links and groups that users joined to perform attribute inference. They found that, with only social links, the approach LINK achieves the best performance. 
LINK represents each user as a binary feature vector, and a feature has a value of 1 if the user is a friend of the person that corresponds to the feature. Then LINK  learns classifiers for attribute inference using these feature vectors. Gong et al.~\cite{gong2014joint} transformed attribute inference to a link prediction problem.
Moreover, they showed that their approaches CN-SAN, AA-SAN, and RWwR-SAN outperform LINK. 

Mislove et al.~\cite{Mislove10} proposed to identify a local community in the social network by taking some seed users that share the same attribute value, and then they predicted all users in the local community to have the shared attribute value. Their approach is not able to infer attributes for users that are not in any local communities.  Moreover, this approach is data dependent since detected communities might not correlate with the attribute value. For instance, Trauda et al.~\cite{Trauda12}  found that communities in a MIT male network are correlated with residence but a female network does not have such property.  

Thomas et al.~\cite{AttributeInferenceFriendPET10} studied the inference of  attributes such as gender, political views, and religious views. They used multi-label classification methods and leveraged features from users' friends and wall posts. Moreover, they proposed the concept of multi-party privacy to defend against attribute inference. 

\myparatight{Behavior-based attribute inference} Weinsberg et al.~\cite{weinsberg2012blurme} investigated the inference of gender using the rating scores that users gave to different movies. In particular, they constructed a feature vector for each user; the $i$th entry of the feature vector is the rating score that the user gave to the $i$th movie if the user reviewed the $i$th movie, otherwise the  $i$th entry is 0. They compared a few  classifiers including Logistic Regression (LG)~\cite{logisticregression}, SVM~\cite{Cortes95}, and Naive Bayes~\cite{McCallum98}, and they found that  LG outperforms the other approaches. Bhagat et al.~\cite{BhagatAttributeInference14} studied attribute inference in an active learning framework. Specifically, they investigated which movies we should ask users to review in order to improve the inference accuracy the most. However,  this approach might not be applicable in real-world scenarios because users might not be interested in reviewing the selected movies.  

Chaabane et al.~\cite{Chaabane12} used the information about the musics users like to infer attributes. 
They augmented the musics  with the corresponding Wikipedia pages and then used topic modeling techniques to identify the latent similarities between musics. A user is predicted to share attributes with those that like similar musics with the user. 
Kosinski et al.~\cite{kosinski2013private} tried to infer various attributes based on the list of pages that users liked on Facebook. Similar to the work performed by Weinsberg et al.~\cite{weinsberg2012blurme}, they constructed a feature vector from the Facebook likes and used Logistic Regression to train classifiers to distinguish users with different attributes. 
Luo et al.~\cite{luo2014you} constructed a model to infer household structures using  users' viewing behaviors in Internet Protocol Television (IPTV) systems, and  they showed promising results. 

\myparatight{Other approaches}  Bonneau et al.~\cite{pryingdataoutofOSN} studied the extraction of private user data in online social networks via various attacks such as account compromise, malicious applications, and fake accounts. 
These attacks can not infer user attributes that users do not provide in their profiles, while our attack can. Otterbacher~\cite{Otterbacher10} studied the inference of gender using users' writing styles. 
Zamal et al.~\cite{attriInfer} used a user's tweets and her neighbors' tweets to infer attributes. They didn't consider social structures nor user behaviors.
Gupta et al.~\cite{Gupta13} tried to infer interests of a Facebook user via  sentiment-oriented mining on the Facebook pages that were liked by the user. Zhong et al.~\cite{AttributeInferenceFromLocationWSDM15} demonstrated the possibility of inferring user attributes using the list of locations where the user has checked in. 
 These studies are orthogonal to ours since they exploited information sources other than the social structures and behaviors that we focus on.  
 
 Attribute inference using social structure and behavior could also be solved by a social recommender system (e.g., \cite{YeSigIR12}). However, such approaches have higher computational complexity than our method for attacking a targeted user, and it is challenging for them to have theoretical guarantees as our attack. For instance, the approach proposed by Ye et al.~\cite{YeSigIR12} has a time complexity of $O(m\cdot k \cdot f \cdot d)$ on a single machine, where $m$ is the number of edges, $k$ is the latent topic size, $f$ is the average number of friends, and $d$ is the number of iterations. Note that both our VIAL and this approach can be parallelized on a cluster.

\section{Conclusion and Future Work}
In this work, we study the problem of attribute inference via combining social structures and user behaviors that are publicly available in online social networks. To this end, we first propose a \emph{social-behavior-attribute (SBA)} network model to gracefully integrate social structures, user behaviors, and their interactions with user attributes. Based on the SBA network model, we design a \emph{vote distribution attack (VIAL)} to perform attribute inference. We demonstrate the effectiveness of our attack both theoretically and empirically. In particular, via empirical evaluations on a real-world large scale dataset with 1.1 million users, we find that attribute inference is a serious practical privacy attack to online social network users and an attacker can successfully attack more users when considering both social structures and  user behaviors. The fundamental reason why our attack succeeds is that private user attributes are statistically correlated with publicly available information, and our attack captures such correlations to map publicly available information to private user attributes. 


 A few interesting directions for future work include learning the link weights of a SBA network, generalizing VIAL to infer hidden social relationships between users, as well as defending against our inference attacks.

 \section{Acknowledgements}
 We would like to thank the anonymous reviewers for their insightful feedback. 
 This work is supported by College of Engineering, 
 Department of Electrical and Computer Engineering of the Iowa State University.

{

}

\appendix

\section{Proof of Theorem~\ref{converge}}
\label{app:converge}
According to  Equation~\ref{iter}, we have:
 \begin{align}
\vec{s}^{(i)}_v = (1-\alpha) ^ i ({\bf M}^T)^i \vec{s}^{(0)}_v  + \alpha (\sum_{k=0}^{i-1} (1-\alpha)^k ({\bf M}^T)^k) \vec{e}_v.
\end{align}

Therefore, 
 \begin{align}
 \lim_{i\to \infty}\vec{s}^{(i)}_v & = \lim_{i\to \infty}\alpha (\sum_{k=0}^{i-1} (1-\alpha)^k ({\bf M}^T)^k) \vec{e}_v \nonumber \\
 &= \alpha(I-(1-\alpha) {\bf M}^T)^{-1} \vec{e}_v.
 \end{align}
We note that the matrix $(I-(1-\alpha) {\bf M}^T)$  is nonsingular because it is strictly diagonally dominant.  

\section{Proof of Theorem~\ref{nobacktracking}}
\label{app:nobacktracking}
The matrix {\bf M} has non-negative entries, and each row of {\bf M} sums to be 1. Therefore, {\bf M} can be viewed as a transition matrix. In particular, {\bf M} can be viewed as a transition matrix of the following Markov chain on the SBA network: each social node is a state of the Markov chain; the transition probability from a social node $u$ to another social node $x$ is ${\bf M}_{ux}$, i.e., a social node $u$ can only transit to its social neighbors or hop-2 social neighbors with non-zero probabilities. 

When the SBA network is connected, the above Markov chain is \emph{irreducible} and
\emph{aperiodic}. Therefore, the Markov chain has a unique stationary distribution $\vec{\pi}$. Moreover, according to the Perron-Frobenius theorem~\cite{PerronFrobeniustheorem}, we have:
 \begin{align}
\lim_{i\to \infty} ({\bf M}^T)^i = [\vec{\pi}\  \vec{\pi} \cdots \vec{\pi}] \nonumber
\end{align}

When $\alpha=0$, we have  $\vec{s}^{(i)}_v = ({\bf M}^T)^i  \vec{s}^{(0)}_v$. Thus, we have 
\begin{align}
\vec{s}_v &= \lim_{i\to \infty} \vec{s}^{(i)}_v  \nonumber \\
&= \lim_{i\to \infty} ({\bf M}^T)^i  \vec{s}^{(0)}_v  \nonumber \\
& = [\vec{\pi}\  \vec{\pi} \cdots \vec{\pi}] \vec{s}^{(0)}_v \nonumber \\
& =  |V_s| \vec{\pi}, \nonumber
\end{align}
where $|V_s|$ is the sum of the entries of $\vec{s}^{(0)}_v$.

\section{Proof of Corollary~\ref{specialcase}}
\label{app:specialcase}
When $w_S = \tau \cdot {d_{u,S}}$, $w_{BS} = \tau \cdot {d_{u,B}}$, and $w_{AS} = \tau \cdot {d_{u,A}}$ for each user $u$,  the Markov chain defined by the transition matrix {\bf M} is a random walk on a weighted graph $G_w=(V_w, E_w)$, which is defined as follows: $V_w=V_s$,  an edge $(u,x)$ in $E_w$ means that $x$ is $u$'s social neighbor or  hop-2 social neighbor in the SBA network, and the weight of the edge $(u,x)\in E_w$ is   $\delta_{ux,S}\cdot {w_{ux}} +\delta_{ux,BS} \cdot {d_{u,B}} \cdot w_B(u,x)  + \delta_{ux,AS} \cdot {d_{u,A}} \cdot w_A(u,x)$. We can verify that, on the graph $G_w$, the weights of all edges that are incident to a node $u$ sum to $d_u$.  Therefore, the stationary distribution $\vec{\pi}$~\cite{randomwalk} of the random walk on $G_w$ is:
\begin{align}
\vec{\pi} = [\frac{d_{u_1}}{D}\ \frac{d_{u_2}}{D} \cdots \frac{d_{u_{|V_s|}}}{D}]^T.
\end{align}

Thus, according to Theorem~\ref{nobacktracking}, we have $\vec{s}_{vu} = |V_s| \frac{d_u}{D}$.
\end{document}